\documentclass{article}

\usepackage{amsmath,amssymb,natbib,graphicx}
\usepackage{algorithm,color,multirow,rotating}

\RequirePackage{hyperref}

\newtheorem{proposition}{Proposition}
\newtheorem{theorem}{Theorem}
\newtheorem{definition}{Definition}
\newtheorem{proof}{Proof}
\newtheorem{corollary}{Corollary}
\newtheorem{claim}{Claim}

\DeclareMathOperator{\Delti}{\mathbf{\Delta}^{-1}}
\DeclareMathOperator{\Delt}{\mathbf{\Delta}}
\DeclareMathOperator{\Sig}{\mathbf{\Sigma}}
\DeclareMathOperator{\Lam}{\mathbf{\Lambda}}
\DeclareMathOperator{\Sigi}{\mathbf{\Sig}^{-1}}
\DeclareMathOperator{\Gam}{\mathbf{\Gamma}}
\DeclareMathOperator{\A}{\mathbf{A}}

\DeclareMathOperator{\V}{\mathbf{V}}
\DeclareMathOperator{\U}{\mathbf{U}}

\DeclareMathOperator{\uvec}{\mathbf{u}}
\DeclareMathOperator{\vvec}{\mathbf{v}}
\DeclareMathOperator{\zvec}{\mathbf{z}}
\DeclareMathOperator{\wvec}{\mathbf{w}}
\DeclareMathOperator{\D}{\mathbf{D}}
\DeclareMathOperator{\Omeg}{\mathbf{\Omega}}
\DeclareMathOperator{\Pmat}{\mathbf{P}}

\DeclareMathOperator{\Q}{\mathbf{Q}}
\DeclareMathOperator{\I}{\mathbf{I}}
\DeclareMathOperator{\R}{\mathbf{R}}
\DeclareMathOperator{\Y}{\mathbf{Y}}
\DeclareMathOperator{\B}{\mathbf{B}}

\DeclareMathOperator{\Bmat}{\mathbf{B}}
\DeclareMathOperator{\X}{\mathbf{X}}

\DeclareMathOperator{\y}{\mathbf{y}}
\DeclareMathOperator{\Lmat}{\mathbf{L}}
\DeclareMathOperator{\Mmat}{\mathbf{M}}
\DeclareMathOperator{\Smat}{\mathbf{S}}
\DeclareMathOperator{\Emat}{\mathbf{E}}
\DeclareMathOperator{\Zmat}{\mathbf{Z}}

\DeclareMathOperator*{\minimize}{\mathrm{minimize}}
\DeclareMathOperator*{\argmin}{\mathrm{argmin}}
\DeclareMathOperator*{\maximize}{\mathrm{maximize}}

\DeclareMathOperator*{\lamv}{\lambda_{\mathbf{v}}}
\DeclareMathOperator*{\lamu}{\lambda_{\mathbf{u}}}

\begin{document}

\title{A Generalized Least Squares Matrix Decomposition}
\author{Genevera I. Allen\footnote{To
    whom correspondence should be addressed.} \\
{\small Department of Pediatrics-Neurology, Baylor College of
  Medicine,} \\ {\small Jan and Dan Duncan Neurological Research
  Institute, Texas Children's Hospital,} \\
{\small \& Departments of Statistics and Electrical and Computer
  Engineering, Rice University.} 
\vspace{4mm} \\
Logan Grosenick \\
{\small Center for Mind, Brain and Computation, Stanford University.} \vspace{4mm} \\
Jonathan Taylor \\ 
{\small Department of Statistics, Stanford University.}}

\maketitle

\begin{abstract}
Variables in many massive high-dimensional data sets are structured,
arising for example from measurements on a regular grid as in imaging
and time series or from spatial-temporal measurements as in climate
studies.  Classical multivariate techniques ignore these structural
relationships often resulting in poor performance.  
We propose a generalization of the singular value decomposition (SVD) and
principal components analysis (PCA)
that is appropriate for massive data sets with structured variables or
known two-way dependencies.
By finding the best low rank approximation of the data with respect to
a transposable quadratic norm, 
our decomposition, 
entitled the {\em Generalized least squares Matrix Decomposition} (GMD),
directly accounts for structural relationships.  As many variables in
high-dimensional settings are often irrelevant or noisy, we
also regularize our matrix decomposition by adding two-way penalties
to encourage sparsity or smoothness.  We develop fast computational
algorithms using our methods to 
perform generalized PCA (GPCA), sparse GPCA, and
functional GPCA on massive data sets.  
Through simulations and a whole brain functional MRI example we 
demonstrate the utility of our methodology for dimension reduction,
signal recovery, and feature selection with high-dimensional
structured data. \\
{\bf Keywords}: matrix decomposition, singular value
  decomposition,
  principal components analysis, sparse PCA,
functional PCA, structured data, neuroimaging.
\end{abstract}

\section{Introduction}

Principal components analysis (PCA), and the singular value decomposition
(SVD) upon which it is based, form the foundation of much of classical
multivariate analysis.  It is well known, however, that with both
high-dimensional data and functional data, traditional PCA can perform
poorly \citep{silverman_1996, jolliffe_spca_2003,
  johnstone_spca_2004}.  Methods enforcing sparsity or smoothness on
the matrix 
factors, such as sparse or
functional PCA, have been shown to consistently recover the factors in these
settings \citep{silverman_1996, johnstone_spca_2004}.  Recently,
these techniques have been extended to regularize both the row and
column factors of a matrix \citep{huang_2009, witten_pmd_2009,
  lee_ssvd_2010}.  All of these methods, however, can fail to capture
relevant aspects of structured high-dimensional data.  In this paper,
we propose a general and flexible framework for PCA that can directly
account for structural dependencies and incorporate two-way
regularization for exploratory analysis of massive structured data
sets.

Examples of high-dimensional structured data in which classical PCA
performs poorly abound in areas of biomedical imaging, environmental
studies, time series and longitudinal studies, and network data.
Consider, for example, functional MRIs (fMRIs) which measure
three-dimensional images of the brain over time with high spatial
resolution.  Each three-dimensional pixel, or voxel, in the image
corresponds to a measure of the bold oxygenation level dependent
(BOLD) response (hereafter referred to as ``activation''), an indirect
measure of neuronal activity at a particular location in the brain.
Often these voxels are vectorized at 
each time point to form a high-dimensional matrix of voxels ($\approx$
10,000 - 100,000) by time 
points ($\approx$ 100 - 10,000) to be 
analyzed by multivariate methods.   Thus, this data exhibits strong 
spatial dependencies among the rows and strong temporal dependencies
among the columns \citep{lindquist_2008, lazar_2008}.  Multivariate
analysis techniques are often applied to fMRI data to
find regions of interest, or spatially contiguous groups of voxels
exhibiting strong co-activation as well as the time
courses associated with these regions of activity
\citep{lindquist_2008}.  Principal 
components analysis and the SVD, however, are rarely used for this
purpose.  Many have noted that the first several principal components
of fMRI data appear to capture spatial and temporal dependencies in
the noise rather than the patterns of brain activation in which
neuroscientists are interested.   Furthermore,
subsequent principal components often exhibit a mixture of noise
dependencies and brain activation such that the signal in the data
remains unclear 
\citep{friston_1999, calhoun_2001, thirion_2003, viviani_fmri_2005}.
This behavior of classical PCA is typical when it is applied to
structured data with low 
signal compared to the noise dependencies induced by the data
structure.

%% Additionally, we note that the signal in fMRI data cannot be captured
%% by the simple column (temporal) and/or row (voxel) means.  Instead,
%% with principal components, one is interested in finding a deviation
%% from a null mean model, that is the factors,$\U$ and 
%% $\V$ are assumed to be fixed, or from a null covariance model, that is $\U$ and
%% $\V$ are assumed to be random.  In other words, one is interested in finding
%% deviations 
%% from a steady state of neuronal activity to identify the regions of
%% activation and their temporal activation patterns.  

To understand the poor performance of the SVD and PCA is these settings,
we examine the model mathematically.  We observe data $\Y \in \Re^{n
  \times p}$, for which 
standard PCA considers the following model: $\Y = \Mmat + \U \D
\V^{T} + \Emat$, 
where $\Mmat$ denotes a mean matrix, $\D$ the singular values,
$\U$ and $\V$ the left and right factors, respectively,
and $\Emat$ the noise matrix.  
 Implicitly, SVD models assume that the elements of
$\Emat$ are independently and identically distributed, as 
the SVD is calculated by minimizing the Frobenius norm loss
function: $|| \X - \U \D \V^{T} ||_{F}^{2}  = \sum_{i=1}^{n}
\sum_{j=1}^{p} ( X_{ij} - \uvec_{i} \D \vvec_{j}^{T} )^{2}$,  where
$\uvec_{i}$ is the $i^{th}$ column of $\U$ and $\vvec_{j}$ is
analogous and  $\X$ denotes the centered data, $\X = \Y - \Mmat$.
This sums of 
squares loss function weights errors associated with each matrix
element equally and cross-product errors,
between elements $X_{ij}$ and $X_{i'j'}$, for example, are ignored.  It
comes as no surprise then that the Frobenius norm loss and thus the
SVD perform poorly with data exhibiting strong structural dependencies
among the matrix elements.  To permit unequal weighting of the matrix
errors according to the data structure, we assume that the noise is
structured: $\mathrm{vec}(\Emat) \sim (\mathbf{0},  \R^{-1} \otimes
\Q^{-1} )$, or 
the covariance of the noise is separable and has a
Kronecker product covariance \citep{gupta_nagar}.  Our loss function,
then changes from the Frobenius norm to a transposable quadratic norm
that permits unequal weighting of the matrix error terms based on
$\Q$ and $\R$: $|| \X - \U \D \V^{T} ||_{\Q,\R}^{2}  = \sum_{i=1}^{n}
   \sum_{j=1}^{p} \sum_{i'=1}^{n}
   \sum_{j'=1}^{p} \Q_{ii'} \R_{jj'} ( X_{ij} - \uvec_{i} \D
   \vvec_{j}^{T} ) ( X_{i'j'} - \uvec_{i'} \D \vvec_{j'}^{T} )$.
By finding the best low rank approximation of the data with respect to
this transposable quadratic norm, 
we develop a decomposition that directly accounts for two-way structural
dependencies in the data.  This gives us a method, {\em Generalized
  PCA}, that extends PCA for applications to structured data.

While this paper was initially under review, previous work on unequal
weighting of matrix elements in PCA via a row and column generalizing
operators came to our attention \citep{escoufier_1977}.  This so
called, ``duality 
approach to PCA'' is known in the French multivariate
 community, but has perhaps not gained the popularity in
the statistics literature that it deserves.    Developed predominately
in the 1970's and 80's for 
applications in ecological statistics, few texts on this were published
in English \citep{escoufier_1977, tenenhaus_young_1985,
  escoufier_1987}.  Recently, a  few works review these methods in 
relation to classical multivariate statistics \citep{escoufier_2006,
  purdom_2007, dray_ade4_2007, holmes_2008}. While these works propose
a mathematical framework for unequal weighting matrices in PCA, much
more statistical and computational development is needed in order to
directly apply these methods to high-dimensional structured data.

In this paper, we aim to develop a flexible mathematical framework for
PCA that accounts for known structure and incorporates regularization
in a manner that is computationally feasible for analysis of massive
structured data sets.  Our specific contributions include (1) a matrix
decomposition that accounts for known two-way structure, (2) a
framework for two-way regularization in PCA and in Generalized PCA,
and (3) results and algorithms allowing one to compute (1) and (2) for
ultra high-dimensional data.  Specifically, beyond the previous work
in this area reviewed by \citet{holmes_2008}, we provide results
allowing for (i) 
general weighting matrices, (ii) computational approaches for
high-dimensional data, and (iii) a framework for two-way
regularization in the context of Generalized PCA.  As our methods are
flexible and general, they will enjoy wide applicability to a variety
of structured data sets common in medical imaging, remote sensing,
engineering, environmetrics, and networking.

The paper is organized as follows.  In Sections 2.1 through 2.5, we
develop the mathematical framework for our {\em Generalized
  Least Squares Matrix Decomposition} and resulting Generalized PCA.
In these sections, we assume that the weighting matrices, or {\em
  quadratic operators} are fixed and known.  Then, by discussing
interpretations of our approach and relations to existing
methodology, we present classes of quadratic operators in Section 2.6
that are appropriate for many structured data sets.  As the focus of
this paper is the statistical methodology for PCA with
high-dimensional structured data, our intent here is to give the
reader intuition on the quadratic operators and leave other aspects
for future work.  In Section ~\ref{section_gpmf}, we propose a general
framework for incorporating two-way regularization in PCA and
Generalized PCA.  We demonstrate how this framework leads to methods
for Sparse and Functional Generalized PCA.    In Section
~\ref{section_results}, we give
results on both simulated data and real functional MRI data.  We
conclude, in Section  
~\ref{section_discussion}, with a discussion of the 
assumptions, implications, and possible extensions of our work.

\section{A Generalized Least Squares Matrix Decomposition}
\label{section_gmd}

We present a generalization of the singular value decomposition (SVD)
that incorporates known dependencies or structural information about
the data. Our methods can be used to perform Generalized PCA for
massive structured data sets.

%% Our novel decomposition, the {\em Generalized least squares Matrix 
%%   Decomposition} (GMD), is based on a transposable quadratic norm
%% described below. 
%% We show that this generalization
%% is simply an SVD of the sphered 
%% data, although one can obtain the GMD solution without
%% computing this sphered data matrix.  In developing this approach, we
%% also reveal the inspiration for the name: 
%% the GMD can be computed by iteratively performing generalized least
%% squares, a modification of the power method algorithm for computing the
%% rank-one SVD.  We also demonstrate how graph Laplacians can be used in
%% conjunction with the GMD when only the structure of the noise is known, and
%% conclude with an illustrative example.

\subsection{GMD Problem}

Here and for the remainder of the paper, we will assume that the data,
$\X$ has previously been centered.
We begin by reviewing the SVD which can be written as $\X = \U \D \V^{T}$
with the data $\X \in \Re^{n \times p}$, and where $\U \in \Re^{n \times p}$
and $\V \in \Re^{p \times p}$ are orthonormal and $\D \in \Re^{p
  \times p}$ is diagonal with non-negative diagonal elements.  Suppose
one wishes to find the best low rank, $K < \mathrm{min}(n,p)$,
approximation to the data.  Recall that
the SVD gives the best rank-$K$
approximation with respect to the Frobenius norm
\citep{eckart_1936}:
\begin{align}
\label{svd}
\minimize_{\U: \mathrm{rank}(\U) = K,  \D: \D \in \mathcal{D},  \V:
  \mathrm{rank}(\V) = K} \ \ &  || \X - \U \D \V^{T} ||_{F}^{2}
  \nonumber \\
\textrm{subject to } \ \ &  \U^{T} \U = \I_{(K)}, \ \   \V^{T} \V =
  \I_{(K)} \ \ \& \ \ \mathrm{diag}(\D) \geq 0. 
\end{align}
Here, $\mathcal{D}$ is the class of diagonal matrices.  
The solution  is given by the first $K$ singular vectors and singular
values of $\X$. 

The Frobenius norm weights all errors equally in
the loss function.  We seek a generalization of
the SVD that allows 
for unequal weighting of the errors to account for structure or
known dependencies in the data.  To this end, we introduce a loss
function given by a transposable quadratic norm, the
$\Q,\R$-norm, defined as follows:
\begin{definition}
Let $\Q \in \Re^{n \times n}$ and $\R \in \Re^{p \times p}$ be
positive semi-definite matrices, $\Q, \R \succeq 0$.  Then the
$\Q,\R$-norm (or semi-norm) is defined as $|| \X ||_{\Q,\R} = \sqrt{
  \mathrm{tr} ( \Q 
  \X \R \X^{T} )}$.  
\end{definition}
We note that if $\Q$ and $\R$ are both positive definite, $\Q,\R \succ
0$, then the $\Q,\R$-norm is a proper matrix norm.  If $\Q$ or $\R$ are
positive semi-definite, then it is a semi-norm, meaning that for $\X
\neq 0$, $|| \X ||_{\Q,\R} = 0$ if $\X \in null(\Q)$ or if $\X^{T} \in
null(\R)$.  We call the $\Q,\R$-norm a transposable quadratic norm, as
it right and left multiplies $\X$ and is thus an extension of the quadratic norm, or 
``norm induced by an inner product space'' \citep{boyd_convex,
  horn_johnson}.    Note that $|| \X ||_{\Q} \triangleq || \X
||_{\Q,\mathbf{I}}$ and $|| \X^{T} ||_{\R} \triangleq || \X^{T}
||_{\R, \mathbf{I}}$ are
then quadratic norms.

The GMD is then taken to be the best
rank-$K$ approximation to the data with respect to the $\Q,\R$-norm:
\begin{align}
\label{gmd}
\minimize_{\U: \mathrm{rank}(\U) = K, \D: \D \in \mathcal{D}, \V: \mathrm{rank}(\V) = K}
\hspace{4mm} &  || \X - 
\U \D \V^{T} ||_{\Q,\R}^{2}  \nonumber \\
\textrm{subject to } \hspace{4mm} & \U^{T} \Q \U = \I_{(K)},
\ \  \V^{T} \R \V = \I_{(K)} \ \ \& \ \ \mathrm{diag}(\D) \geq 0.
\end{align}
So we do not confuse the elements of the GMD with that of the SVD, we
call $\U$ and $\V$ the left 
and right GMD factors respectively, and the diagonal elements of $\D$,
the GMD values.  The matrices $\Q$ and $\R$ are called the left and
right quadratic operators of the GMD respectively.
Notice that the left and right GMD factors
are constrained to be orthogonal in the inner product space induced by
the $\Q$-norm and $\R$-norm respectively.  Since the Frobenius
norm is a special case 
of the $\Q,\R$-norm, the GMD is in fact a generalization of the SVD.

We pause to motivate the rationale for finding a matrix decomposition
with respect to the $\Q,\R$-norm.  Consider a matrix-variate Gaussian
matrix, $\X \sim 
N_{n,p}( \U \D \V^{T} , \Q^{-1},\R^{-1})$, or $\mathrm{vec}(\X) \sim
N(\mathrm{vec}( \U \D \V^{T} ), \R^{-1} \otimes \Q^{-1})$ in terms of
the familiar multivariate normal.  The 
normal log-likelihood can be written as: 
\begin{align*}
\ell(\X | \Q^{-1}, \R^{-1} ) &\propto   \mathrm{tr}    \left( \Q (\X -
\U \D \V^{T}) \R    (\X - \U \D \V^{T})^{T} \right)  =    ||    \X  -
\U \D \V^{T} ||_{\Q,\R}^{2}.  
\end{align*}
Thus, as the Frobenius norm loss of the SVD is proportional to the
log-likelihood of the spherical Gaussian with mean $\U \D \V^{T}$, the
$\Q,\R$-norm loss is proportional to the log-likelihood of the
matrix-variate Gaussian with mean $\U \D \V^{T}$, row covariance
$\Q^{-1}$, and column covariance $\R^{-1}$.  
In general, using the $\Q,\R$-norm loss assumes that the
covariance of the data arises from the Kronecker product between the
row and column covariances, or that the
covariance structure is separable.  Several statistical tests exist to
check these assumptions with real data \citep{mitchell_2006,
  li_genton_2008}.   Hence, if the dependence structure of the data is
known, taking the quadratic operators to be the inverse
covariance or precision matrices directly accounts for two-way
dependencies in the data.  We assume that $\Q$ and $\R$ are
known in the development of the GMD and discuss possible choices for
these with structured data in Section \ref{section_interpret}.

%% We, then, define the GMD to be the best
%% rank-$K$ approximation to the data with respect to the $\Q,\R$-norm:
%% \begin{align}
%% \label{gmd}
%% \minimize_{\U: \mathrm{rank}(\U) = K, \D: \D \in \mathcal{D}, \V: \mathrm{rank}(\V) = K}
%% \hspace{4mm} &  || \X - 
%% \U \D \V^{T} ||_{\Q,\R}^{2}  \nonumber \\
%% \textrm{subject to } \hspace{4mm} & ||\U ||_{\Q} = \I_{(K)},
%% \ \  ||
%% \V ||_{\R} = \I_{(K)} \ \ \& \ \ \mathrm{diag}(\D) \geq 0.
%% \end{align}
%% So we do not confuse the elements of the GMD with that of the SVD, we
%% call $\U$ and $\V$ the left 
%% and right GMD factors respectively, and the diagonal elements of $\D$,
%% the GMD values.  Notice that the left and right GMD factors
%% are constrained to be orthogonal in the inner product space induced by
%% the $\Q$-norm and $\R$-norm respectively, instead of the Frobenius norm
%% as in the SVD.   

\subsection{GMD Solution}
\label{section_gmd_sol}

The GMD solution, $\hat{\X} = \U^{*} \D^{*} (\V^{*})^{T}$, is
comprised of 
$\U^{*}, \D^{*}$ and $\V^{*}$, the optimal points minimizing the GMD
problem \eqref{gmd}.  
The following result states that the GMD solution is an SVD of an
altered data matrix. 
\begin{theorem}
\label{gmd_sol}
Suppose $\mathrm{rank}(\Q) = l$ and $\mathrm{rank}(\R) = m$.
Decompose $\Q$ and $\R$ by letting $\Q = \tilde{\Q} \tilde{\Q}^{T}$
and $\R = \tilde{\R} \tilde{\R}^{T}$ where $\tilde{\Q} \in \Re^{n
  \times l}$ and $\tilde{\R} \in \Re^{p \times m}$ are of full column
rank.  Define $\tilde{\X} = \tilde{\Q}^{T} \X \tilde{\R}$ and let
$\tilde{\X} = \tilde{\U} \tilde{\D} \tilde{\V}^{T}$ be the SVD of
$\tilde{\X}$.  Then, the GMD solution, $\hat{\X} = \U^{*} \D^{*}
(\V^{*})^{T}$,  is given by the GMD factors
$\U^{*} = \tilde{\Q}^{-1} \tilde{\U}$ and $\V^{*} = \tilde{\R}^{-1}
\tilde{\V}$ and the GMD values, $\D^{*} = \tilde{\D}$.  Here,
$\tilde{\Q}^{-1}$ and $\tilde{\R}^{-1}$ are any left matrix inverse:
$(\tilde{\Q}^{-1})^{T} \tilde{\Q} = \mathbf{I}_{(l)}$ and $(
\tilde{\R}^{-1})^{T} \tilde{\R} = \mathbf{I}_{(m)}$. 
\end{theorem}

We make some brief comments regarding this
result.  First, the decomposition, $\Q = \tilde{\Q} \tilde{\Q}^{T}$
where $\tilde{\Q} \in \Re^{n \times l}$ is of full column rank, exists since
$\Q \succeq 0$ \citep{horn_johnson}; the decomposition for $\R$ exists
similarly.  A possible form of this decomposition, the resulting
$\tilde{\X}$, and the GMD solution $\hat{\X}$ can be obtained from the
eigenvalue decomposition of $\Q$ and $\R$: $\Q = \Gam_{\Q} \Lam_{\Q}
\Gam_{\Q}^{T}$ and $\R = \Gam_{\R} \Lam_{\R} \Gam_{\R}^{T}$.  If $\Q$
is full rank, then we can take $\tilde{\Q} = \Q^{1/2} = \Gam_{\Q} 
\Lam_{\Q}^{1/2} \Gam_{\Q}^{T}$, giving the GMD factor $\U^{*} =
\Gam_{\Q} \Lam_{\Q}^{-1/2} \Gam_{\Q}^{T} \tilde{\U}$, and similarly
for $\R$ and $\V$.  On the other 
hand, if $\Q \succeq 0$, then a possible value for $\tilde{\Q}$ is
$\Gam_{\Q}(\cdot,1:l) \Lam^{1/2}_{\Q}(1:l,1:l)$, giving an $n \times l$
matrix with full column rank.  The GMD factor, $\U^{*}$ is then given
by $\Gam_{\Q}(\cdot, 1:l) \Lam_{\Q}^{-1/2}(1:l,1:l) \tilde{\U}$.

From Theorem ~\ref{gmd_sol}, we see that the GMD solution can be
obtained from the SVD of $\tilde{\X}$.  Let us assume for a moment
that $\X$ is matrix-variate Gaussian with row and column covariance
$\Q^{-1}$ and $\R^{-1}$ respectively.  Then, the GMD is like taking
the SVD of the {\em sphered} data, as right and left multiplying by
$\tilde{\Q}$ and $\tilde{\R}$ yields data with identity covariance.
In other words, the GMD de-correlates the data so that the SVD with
equally weighted errors is appropriate.  While the GMD values are the
singular values of this sphered data, the covariance is multiplied
back into the GMD factors.

This relationship to the matrix-variate normal also begs the question,
if one has data with row and column correlations, why not take the SVD
of the two-way sphered or ``whitened'' data?  This is inadvisable for
several reasons.  First, pre-whitening the data and then taking the
SVD yields matrix factors that are in the wrong basis and are thus
uninterpretable in the original data space.  Given this, one may
advocate pre-whitening, taking the SVD and then re-whitening, or in
other words multiplying the correlation back in to the SVD factors.
This approach, however, is still problematic.  In the special case
where $\Q$ and $\R$ are of full rank, the GMD solution is exactly to
same as this pre and re-whitening approach.  In the general case where
$\Q$ and $\R$ are positive semi-definite, however, whitening cannot be
directly performed as the covariances are not full rank.  In the
following section, we will show that the GMD solution can be
computed without taking any matrix inverses, square roots or
eigendecompositions and is thus computationally more attractive than
naive whitening.  Finally, as our eventual goal is to develop a
framework for both structured data and two-way regularization, naive
pre-whitening and then re-whitening of the data would destroy any
estimated sparsity or smoothness from regularization methods and is
thus undesirable.  Therefore, the
GMD solution given in Theorem ~\ref{gmd_sol} is the mathematically
correct way to ``whiten'' the data in the context of the SVD and is
superior to a naive whitening approach.

Next, we explore some of the mathematical properties of our GMD solution,
$\hat{\X} = \U^{*} \D^{*} (\V^{*})^{T}$ in the following corollaries:  
\begin{corollary}
\label{cor_global}
The GMD solution is the global minimizer of the GMD problem,
\eqref{gmd}. 
\end{corollary}
\begin{corollary}
\label{cor_k}
Assume that $range(\Q) \cap  null(\X) = \emptyset$ and $range(\R) \cap
null (\X) = \emptyset$.   Then,
there exists at most $k = \mathrm{min}\{ \mathrm{rank}(\X),
\mathrm{rank}(\Q), \mathrm{rank}(\R) \}$ non-zero GMD values and
corresponding left and right GMD factors.  
\end{corollary}
\begin{corollary}
\label{cor_error}
With the assumptions and $k$ defined as in Corollary ~\ref{cor_k}, the
rank-$k$ GMD solution has zero
reconstruction error in the $\Q,\R$-norm.  If in addition, $k =
\mathrm{rank}(\X)$ and $\Q$ and $\R$ are full rank, then $\X \equiv
\U^{*} \D^{*} (\V^{*})^{T}$, that is, the GMD is an exact matrix
decomposition.  
\end{corollary}
\begin{corollary}
\label{cor_unique}
The GMD values, $\D^{*}$, are unique up to multiplicity. 
If in addition, $\Q$ and $\R$
are full rank and the non-zero GMD values are distinct, 
then the GMD factors $\U^{*}$ and $\V^{*}$ corresponding to the
non-zero GMD values are essentially unique (up
to a sign change) and the GMD factorization is unique.  
\end{corollary}

Some further comments on these results are warranted.  In particular,
Theorem ~\ref{gmd_sol} is straightforward and less interesting
when $\Q$ and $\R$ are full rank, the case discussed in
\citep{escoufier_2006, purdom_2007, holmes_2008}.
When the quadratic 
operators are positive semi-definite, however, the fact that a global
minimizer to the GMD problem, which is non-convex, that has a closed
form can be obtained is 
not immediately clear.  The result stems from the relation of the GMD
to the SVD and the fact that the latter is a unique matrix
decomposition in a lower dimensional subspace defined in Corollary
\ref{cor_k}.    Also note that when $\Q$ and $\R$ are full rank the
GMD is an exact matrix 
decomposition; in the alternative scenario, we do not recover $\X$
exactly, but 
obtain zero reconstruction error in the $\Q,\R$-norm.  Finally, we
note that when $\Q$ and $\R$ are rank deficient, there are possibly
many optimal points for the GMD factors, although the resulting GMD
solution is still a global minimizer of the GMD problem.
In Section 2.6, we will see that permitting the quadratic operators to
be positive semi-definite allows for much greater flexibility when
modeling high-dimensional structured data.

\subsection{GMD Algorithm}

While the GMD solution given in Theorem ~\ref{gmd_sol} is
conceptually simple, its computation, based on computing $\tilde{\Q}$,
$\tilde{\R}$ and the SVD of
$\tilde{\X}$,  is infeasible for massive data
sets common in neuroimaging.  We seek a method of obtaining the 
GMD solution that avoids computing and storing $\tilde{\Q}$ and
$\tilde{\R}$ and thus permits our methods to be used with massive
structured data sets.  

\begin{algorithm}[!!h]
\caption{GMD Algorithm (Power Method)}
\label{gmd_alg}
\begin{enumerate}
\item Let $\hat{\X}^{(1)} = \X$ and initialize $\uvec_{1}$ and
  $\vvec_{1}$.  
\item For $k = 1 \ldots K$:
\begin{enumerate}
\item Repeat until convergence:
\begin{itemize}
\item Set $\uvec_{k} = \frac{\hat{\X}^{(k)} \R
  \vvec_{k}}{||\hat{\X}^{(k)} \R  \vvec_{k}||_{\Q}}$. 
\item Set $\vvec_{k} = \frac{(\hat{\X}^{(k)})^{T} \Q
  \uvec_{k}}{||(\hat{\X}^{(k)})^{T} \Q  \uvec_{k}||_{\R}}$. 
\end{itemize}
\item Set $d_{k} = \uvec_{k}^{T} \Q \hat{\X}^{(k)} \R \vvec_{k}$. 
\item Set $\hat{\X}^{(k+1)} = \hat{\X}^{(k)} - \uvec_{k} d_{k}
  \vvec_{k}^{T}$. 
\end{enumerate}
\item Return $\U^{*} = [ \uvec_{1}, \ldots, \uvec_{K}]$, $\V^{*} = [ \vvec_{1},
  \ldots \vvec_{K} ]$ and $\D^{*} = \mathrm{diag}(d_{1}, \ldots d_{K})$.
\end{enumerate}
\end{algorithm}

\begin{proposition}
\label{gmd_alg_sol}
If $\hat{\uvec}$ and $\hat{\vvec}$ are initialized such that
$\hat{\uvec}^{T} \Q \uvec^{*} \neq 0$ and $\hat{\vvec}^{T} \Q
\vvec^{*} \neq 0$, then
Algorithm ~\ref{gmd_alg} converges to the GMD solution given in
Theorem ~\ref{gmd_sol}, the global minimizer of the GMD problem. If in
addition, $\Q$ and 
$\R$ are full rank, then it converges to the unique global solution.
\end{proposition}

In Algorithm \ref{gmd_alg}, we give a method of computing the
components of the GMD solution that is a variation of the familiar
power method for calculating the SVD \citep{golub_van_loan_1996}.
Proposition \ref{gmd_alg_sol} states that the GMD Algorithm based
on this power method
converges to the GMD solution.  Notice that we do not need to calculate
$\tilde{\Q}$ and $\tilde{\R}$ and thus, this 
algorithm is less computationally
intensive for high-dimensional data sets than finding the solution via
Theorem \ref{gmd_sol} or the computational approaches given in
\citet{escoufier_1987} for positive definite operators.

At this point, we pause to discuss the name of our matrix
decomposition.  Recall that 
the power method, or alternating least squares method, sequentially
estimates  $\uvec$ with $\vvec$ fixed and vice versa by solving
least squares problems and then re-scaling the estimates.  Each step
of the GMD algorithm, then, estimates the factors by solving the
following generalized least squares problems and re-scaling: $|| \X -
(d' \uvec') \vvec^{T} ||_{\Q}^{2}$ and $|| \X^{T} - (d'   \vvec')
\uvec^{T} ||_{\R}^{2}$. This is then the inspiration for the name of
our matrix decomposition.

\subsection{Generalized Principal Components}

In this section we show that the GMD can be used to perform
Generalized Principal Components Analysis (GPCA).   Note that this
result was first shown in \citet{escoufier_1977} for positive
definite generalizing operators, but we review it here for
completeness.  The results in the previous 
section allow one to compute these GPCs
for high-dimensional data when the quadratic operators may not be
full rank and we wish to avoid computing SVDs and eigenvalue
decompositions.

Recall that the SVD problem can be written as finding the linear
combination of variables maximizing the sample variance such that
these linear combinations are orthonormal.  
For GPCA, we seek to project the data onto
the linear combination of variables such that the sample variance is
maximized in a
space that accounts for the structure or dependencies in the data.  
By transforming all inner product spaces to those induced by the
$\Q,\R$-norm, we arrive at the following Generalized PCA
problem: 
\begin{align}
\label{gpca}
\maximize_{\vvec_{k}} \ \ \vvec_{k}^{T} \R \X^{T} \Q \X \R \vvec_{k}
\ \  \mathrm{subject \hspace{1mm} to} \ \  \vvec_{k}^{T} \R
\vvec_{k} = 1 \  \&  \ \vvec_{k}^{T} \R \vvec_{k'} = 0 \ \ \forall \
\ k'<k.
\end{align}
Notice that the loadings of the GPCs are orthogonal in the
$\R$-norm.  The $k^{th}$ GPC is then given by $\zvec_{k} = \X \R \vvec_{k}$.

\begin{proposition}
\label{prop_gpca}
The solution to the $k^{th}$ Generalized principal component problem,
\eqref{gpca}, is given by the $k^{th}$ right GMD factor. 
\end{proposition}

\begin{corollary}
\label{cor_gpca_var_ex}
The proportion of variance explained by the $k^{th}$ Generalized
principal component is given by $d_{k}^{2} / || \X ||_{\Q,\R}^{2}$.
\end{corollary}

Just as the SVD can be used to find the principal components, the GMD
can be used to find the 
generalized principal components.  Hence, GPCA can be performed using
the GMD algorithm which  does not require
calculating $\tilde{\Q}$ and $\tilde{\R}$.  This allows one to
efficiently perform GPCA for massive structured data sets.

\subsection{Interpretations, Connections \& Quadratic Operators}
\label{section_interpret}

In the previous sections, we have introduced the methodology for the
Generalized Least Squares Matrix Decomposition and have outlined how
this can be used to compute GPCs for high-dimensional
data.  Here, we pause to discuss some interpretations of our methods
and how they relate to existing methodology for non-linear dimension
reduction.  These interpretations and
relationships will help us understand the role of the quadratic
operators and the classes of quadratic operators that may be appropriate
for types of structured high-dimensional data.  As the focus
of this paper is the statistical methodology of PCA for
high-dimensional structured data, we leave much of the study of
quadratic operators for specific applications as future work.

First, there are many connections and interpretations of the GMD in
the realm of classical matrix analysis and multivariate analysis.  We
will refrain from listing all of these here, but note that there are
close relations to the GSVD of \citet{van_loan_gsvd_1976} and generalized
eigenvalue problems \citep{golub_van_loan_1996} as well as statistical
methods such as 
discriminant analysis, canonical correlations analysis, and
correspondence analysis.  Many of these connections are discussed in
\citet{holmes_2008} and \citet{purdom_2007}.  Recall also that the GMD
can be interpreted as a maximum likelihood problem for the
matrix-variate normal with row and column inverse covariances fixed and
known.  This connection yields two interpretations worth
noting.  First, the GMD is an extension of the SVD to allow for
heteroscedastic (and separable) row and column errors.  Thus, the
quadratic operators act as weights in the matrix decomposition to
permit non i.i.d. errors.  The second relationship is to whitening the
data with known row and column covariances prior to dimension
reduction.  As discussed in Section ~\ref{section_gmd_sol}, our
methodology offers 
the proper mathematical context for this whitening with many advantages
over a naive whitening approach.

Next, the GMD can be interpreted as decomposing a covariance operator.
Let us assume that the data follows a simple model: 
$\X = \Smat + \Emat$, where $\Smat = \sum_{k=1}^{K} \phi_{k} \uvec_{k}
\vvec_{k}^{T}$ with the factors $\vvec_{k}$ and $\uvec_{k}$ fixed but
unknown and with
the amplitude $\phi_{k}$ random, and $\Emat \sim N_{n,p}( 0, \Q^{-1},
\R^{-1} )$.  Then the (vectorized) covariance of the data can be
written as:
\begin{align*}
\mathrm{Cov}( \mathrm{vec}( \X) ) = \sum_{k=1}^{K} \mathrm{Var}(
\phi_{k} ) \vvec_{k} \vvec_{k}^{T} \otimes \uvec_{k} \uvec_{k}^{T} +
\R^{-1} \otimes \Q^{-1},
\end{align*}
such that $\V^{T} \R \V = \mathbf{I}$ and $\U^{T} \Q \U =
\mathbf{I}$.  In other words, the GMD decomposes the covariance of the
data into a signal component and a noise component such that the
eigenvectors of the signal component are orthonormal to those of the
noise component.  This is an important interpretation to consider when
selecting quadratic operators for particular structured data sets,
discussed subsequently.

Finally, there is a close connection between the GMD and smoothing
dimension reduction methods.  Notice that from Theorem ~\ref{gmd_sol},
the GMD factors $\U$ and
$\V$ will be as smooth as as the smallest eigenvectors corresponding to
non-zero eigenvalues of $\Q$ and
$\R$ respectively.  Thus, if
$\Q$ and $\R$ are taken as smoothing operators, then the GMD can yield
smoothed factors.

%oultine here:
%1) Model based \Q and \R - inverse of a time series process, random
%field covariance structure, if multiple matrices - could estimate via
%matri-variate methods of allen et al.
%2) Graph-based methods - Laplacians - illustrate close connection to 
%3) Smoothing matrices - kernel smoothers

Given these interpretations of the the GMD, we consider classes of
quadratic operators that may be appropriate when applying our methods
to structured data:
\begin{enumerate}
\item Model-based and parametric operators.  The quadratic operators
  could be taken as the inverse covariance matrices implied by common
  models employed with structured data.  These include well-known time
  series autoregressive and moving average processes, random field
  models used for spatial data, and Gaussian
  Markov random fields.  
\item Graphical operators.  As many types of structured data can be
  well represented by a graphical model, the graph Laplacian operator,
  defined as the difference between the degree and adjacency matrix,
  can be used as a quadratic operator. Consider, for example, image
  data sampled on a regular mesh grid; a lattice graph connecting all
  the nearest neighbors well represents the structure of this data.  
\item Smoothing / Structural embedding operators.  Taking quadratic
  operators as smoothing 
  matrices common in functional data analysis will yield smoothed GMD
  factors.  In these smoothing matrices, two variables are typically
  weighted inversely proportional to the distance between them.  Thus,
  these can also be thought of as structural embedding matrices,
  increasing the weights in the loss function between variables that are
  close together on the structural manifold.
\end{enumerate}

While this list of potential classes of quadratic operators is most
certainly incomplete, these provide the flexibility to model many types of
high-dimensional structured data.  Notice that many examples of
quadratic operators in each of these classes are closely related.
Consider, for example, regularly spaced ordered variables as often
seen in time series.  A graph Laplacian of a nearest neighbor graph is
tri-diagonal, is nearly equivalent except for boundary points to the
squared second differences matrix often used for inverse smoothing in
functional data analysis \citep{ramsay_2006}, and has the same zero
pattern as the 
inverse covariance of an autoregressive and moving average order one
process \citep{shaman_1969, galbraith_1974}.  Additionally, the zero
patterns in the inverse 
covariance matrix of Gaussian Markov random fields are closely
connected to Markov networks and undirected graph Laplacians
\citep{rue_2005}.  A 
recent paper, \citet{gaussian_spde_2011}, shows how common stationary
random field 
models specified by their covariance function such as the Mat\'ern
class can be derived from functions of a graph Laplacian.  Finally, 
notice that graphical operators such as Laplacians and smoothing
matrices are typically not positive definite.  Thus, our framework
allowing for general quadratic operators in Section
\ref{section_gmd_sol} opens more possibilities than diagonal
weighting matrices \citep{dray_ade4_2007} and those used in ecological
applications \citep{escoufier_1977, tenenhaus_young_1985}.

There are many further connections of the GMD when employed with
specific quadratic 
operators belonging to the above classes and other recent methods for
non-linear dimension reduction.  Let us first consider the
relationship of the GMD to Functional PCA (FPCA) and the more recent two-way
FPCA.  \citet{silverman_1996} first showed that for
discretized functional data, FPCA can be obtained by
half-smoothing; \citet{huang_2009} elegantly extended this to two-way
half-smoothing.  Let us consider the latter where row and column smoothing
matrices, $\Smat_{u} = (\mathbf{I} + \lambda_{u} \Omeg_{u})^{-1}$ and
$\Smat_{v} = (\mathbf{I} + \lambda_{v} \Omeg_{v})^{-1}$ respectively, are
formed with $\Omeg$ being a penalty matrix such as to give the squared
second differences for example \citep{ramsay_2006} related to the 
structure of the row and column variables.  Then, \citet{huang_2009}
showed that two-way FPCA can be performed by two-way half smoothing: (1)
take the SVD of $\X' = \Smat_{u}^{1/2} \X \Smat_{v}^{1/2}$, then (2)
the FPCA factors are $\U = \Smat_{u}^{1/2} \U'$ and $\V =
\Smat_{v}^{1/2} \V'$.   In other words, a half smoother is applied to
the data, the SVD is taken and another half-smoother is applied to the
SVD factors.  This procedure is quite similar to the GMD
solution outlined in Theorem \ref{gmd_sol}.  Let us assume that $\Q$
and $\R$ are smoothing matrices.  Then, our GMD solution is
essentially like half-smoothing the data, taking the SVD then inverse
half-smoothing the SVD factors.  Thus, unlike two-way FPCA, the GMD
does not half-smooth both the data and the factors, but only the data
and then finds factors that are in contrast to the smoothed data.  If
the converse is true and we take $\Q$ and $\R$ to be ``inverse
smoothers'' such as a graph Laplacian, then this inverse smoother is
applied to the data, the SVD is taken and the resulting factors are
smoothed.  Dissimilar to two-way FPCA which performs two smoothing
operations, the GMD essentially applies one smoothing operation and
one inverse smoothing operation.

The GMD also shares similarities to other non-linear dimension
reduction techniques such as Spectral Clustering, Local
Linear Embedding, Manifold Learning, and Local Multi-Dimensional
Scaling.  
Spectral clustering, Laplacian embedding, and manifold learning all
involve taking an eigendecomposition of a Laplacian matrix, $\Lmat$,
capturing 
the distances or relationships between variables.  The eigenvectors
corresponding to the smallest eigenvalues are of interest as these
correspond to minimizing a weighted sums of squares: $\mathbf{f}^{T}
\Lmat \mathbf{f} = \sum_{i} \sum_{i'} L_{i,i'} (f_{i} - f_{i'}
)^{2}$, ensuring that the distances between neighboring variables in
the reconstruction $\mathbf{f}$ are small.  Suppose $\Q = \Lmat$ and
$\R = \mathbf{I}$, then the GMD minimizes: $\sum_{i} \sum_{i'}
L_{i,i'} ( \X_{i} - \uvec_{i} \D \V^{T} )^{2}$, similar to the
Laplacian embedding criterion.  Thus, the GMD with quadratic operators
related to the structure or distance between variables ensures that the
errors between the data and its low rank approximation are smallest for
those variables that are close in the data structure.  This concept is
also closely related to Local Multi-Dimensional Scaling which weights
MDS locally by placing more weight on variables that are in close proximity
\citep{chen_lmds_2009}.

The close connections of the GMD to other recent non-linear dimension
reduction techniques provides additional context for the role of the
quadratic operators.  Specifically, these examples indicate that for
structured data, it is often not necessary to directly estimate the quadratic
operators from the data.  If the quadratic operators are taken
as smoothing matrices, structural embedding matrices or Laplacians
related to 
the distances between 
variables, then the GMD has similar properties to many other
non-linear unsupervised learning methods.  In summary, when various quadratic
operators that encode data structure such as the distances between
variables are employed, the GMD can be interpreted as (1) finding the
basis vectors 
of the covariance that are separate (and orthogonal) to that of the
data structure, (2) finding factors that are smooth with respect to
the data structure, or (3) finding an approximation to the data that
has small reconstruction error with respect to the structural
manifold.  Through simulations in
Section \ref{section_simulations}, we will explore the performance of the GMD
for different combinations of graphical and smoothing operators
encoding the known data structure.  As there are many examples of
structured high-dimensional data (e.g. imaging, neuroimaging, microscopy,
hyperspectral imaging, chemometrics, remote sensing, environmental
data, sensor data, and network data), these classes of quadratic
operators provide a wide range of potential applications for our
methodology.

\section{Generalized Penalized Matrix Factorization}
\label{section_gpmf}

With high-dimensional data, many have advocated regularizing principal
components by either automatically selecting relevant features as with
Sparse PCA or by smoothing the factors as with Functional PCA
\citep{silverman_1996, jolliffe_spca_2003, zou_spca_2006, shen_spca_2008,
   huang_2009, witten_pmd_2009, lee_ssvd_2010}.
Regularized PCA can be important for massive structured data as well.
Consider spatio-temporal fMRI data, for example, where many spatial
locations or voxels in the brain are inactive and the time courses are
often extremely noisy.  Automatic feature selection of relevant voxels
and smoothing of the time series in the context of PCA for structured
data would thus be an important contribution. 
In this section, we seek a framework for regularizing the factors of
the GMD by placing penalties on each factor.    In
developing this framework, we reveal an important result demonstrating
the general class of penalties that can be placed on matrix factors that are
to be estimated via deflation: the penalties must
be norms or semi-norms.

\subsection{GPMF Problem}

Recently, many have suggested regularizing the factors of
the SVD by forming two-way penalized regression problems
\citep{huang_2009, witten_pmd_2009, lee_ssvd_2010}.  We briefly review
these existing approaches to understand how to frame a problem that
allows us to place general sparse or smooth penalties on the GMD
factors.

We compare the optimization problems of these three approaches for
computing a single-factor two-way regularized matrix factorization:
\begin{align*}
\textrm{\citet{witten_pmd_2009}}: \  & \maximize_{\vvec, \uvec} \   \uvec^{T}
\X \vvec  \  \mathrm{subject  \hspace{1mm} to} \   \uvec^{T}
\uvec \leq 1, \  \vvec^{T} \vvec 
  \leq 1, \  P_{1}( \uvec ) \leq c_{1}, \ \& \ P_{2}( \vvec ) \leq
  c_{2}. \\
\textrm{\citet{huang_2009}}: \ &  \maximize_{\vvec, \uvec} \  \uvec^{T} \X
  \vvec  - \frac{\lambda}{2} P_{1}(\uvec) P_{2}( \vvec). \\
\textrm{\citet{lee_ssvd_2010}}: \ &  \maximize_{\vvec, \uvec} \  \uvec^{T} \X
  \vvec  - \frac{1}{2} \uvec^{T} \uvec \vvec^{T} \vvec -
  \frac{\lamu}{2} P_{1}(\uvec)  - \frac{\lamv}{2} P_{2}( \vvec).   
\end{align*}
Here, $c_{1}$ and $c_{2}$ are constants, $P_{1}()$ and
$P_{2}()$ are penalty functions, and $\lambda$, $\lamv$ and $\lamu$
are penalty parameters.  These optimization problems
are attractive as they are bi-concave in $\uvec$ and $\vvec$, meaning that
they are concave in $\uvec$ with $\vvec$ fixed and vice versa.  Thus, a
simple maximization strategy of iteratively maximizing with respect to
$\uvec$ and $\vvec$ results in a monotonic ascent algorithm converging
to a local maximum.

These methods, however, differ in the types of penalties that can be
employed and their scaling of the matrix factors.  
\citet{witten_pmd_2009} explicitly restrict the norms of the factors,
while the constraints in the method of \citet{huang_2009} are
implicit because of the types of functional data penalties employed.
Thus, for these methods, as the constants $c_{1}$ and $c_{2}$ or the
penalty parameter, $\lambda$ approach zero, the SVD is returned.  This
is not the case, however, for problem in \citet{lee_ssvd_2010} where
the factors are not constrained in the optimization problem, although
they are later scaled in their algorithm.  Also,
restricting the scale of the factors avoids possible numerical
instabilities when employing the iterative estimation
algorithm (see especially the supplementary materials of
\citet{huang_2009}).  Thus, we prefer the former two approaches for
these reasons.  In \citet{witten_pmd_2009}, however, only the lasso
and fused lasso penalty functions are employed and it is unclear
whether other penalties may be used in their framework.
\citet{huang_2009}, on the other hand, limit their consideration to
quadratic functional data penalties, and their optimization problem
does not implicitly scale the factors for other types of penalties.

As we wish to employ general penalties, specifically sparse and smooth
penalties, on the matrix factors, we adopt an optimization problem
that leads to simple solutions for the factors with a wide class of
penalties, as discussed in the subsequent section.  We then, define
the single-factor {\em Generalized Penalized Matrix Factorization}
(GPMF) problem as the following: 
\begin{align}
\label{gpmf}
\maximize_{\vvec, \uvec} \   \uvec^{T} \Q \X \R
\vvec - \lamv P_{1}( \vvec )  - \lamu P_{2}( \uvec)  \ 
\mathrm{subject \hspace{1mm} to} \  \uvec^{T} \Q \uvec \leq 1
\  \& \  \vvec^{T} \R \vvec \leq 1,
\end{align}
where, as before, $\lamv$ and $\lamu$ are penalty parameters and
$P_{1}()$ and $P_{2}()$ are penalty functions.  
Note that if $\lamu = \lamv = 0$, then 
 the left and right GMD factors can be found from \eqref{gpmf}, as
 desired.  Strictly speaking, 
this problem is the Lagrangian of the problem introduced in
\citet{witten_pmd_2009}, keeping in mind that  we should interpret the
inner products as those 
induced by the $\Q,\R$-norm.  As we will see in Theorem
~\ref{gpmf_sol} in the next section, this problem is tractable for
many common choices of penalty functions and avoids the scaling
problems of other approaches.

%% Strictly speaking, \eqref{pmd} is the Lagrangian for the
%% problem described in \citet{witten_pmd_2009}. For reasons
%% described later, we prefer this Lagrangian form. Another penalty
%% proposed in the literature \cite{huang_2009} considers the
%% following penalty
%% \begin{align}
%% \maximize_{\vvec, \uvec} \ \  \uvec^{T} \X \vvec - \lamv
%%   \lamu P_{1}( \vvec ) P_{2}( \uvec ) \ \ \ \mathrm{subject
%%   \hspace{1mm} to} \ \  \uvec^{T} \uvec \leq 1
%% \ \ \& \ \ \vvec^{T} \vvec \leq 1.
%% \end{align}
%% This penalty is proposed to avoid possible scaling problems with
%% the Lagrangian form \eqref{pmd}. We will see, however, in Theorem \ref{gpmf_sol} that
%% for many common choices of penalties, these scaling problems
%% are not an issue for \eqref{pmd}.

%% We extend the optimization problem \eqref{pmd} to the inner
%% product space induced by the $\Q,\R$-norm, thus generalizing this
%% problem to be appropriate for dependent data.  The single-factor GPMF
%% is defined as follows:
%% \begin{align}
%% \label{gpmf}
%% \maximize_{\vvec, \uvec} \   \uvec^{T} \Q \X \R
%% \vvec - \lamv P_{1}( \vvec )  - \lamu P_{2}( \uvec)  \ 
%% \mathrm{subject \hspace{1mm} to} \  \uvec^{T} \Q \uvec \leq 1
%% \  \& \  \vvec^{T} \R \vvec \leq 1.
%% \end{align}
%% Note that as in \eqref{pmd}, if $\lamu = \lamv = 0$, then \eqref{gpmf}
%% gives us the left and right GMD factors.  Similarly, the GPMF problem
%% is also bi-concave in $\uvec$ and $\vvec$.  

\subsection{GPMF Solution}

We solve the GPMF criterion,
\eqref{gpmf}, via block coordinate ascent by alternating
maximizing with respect to $\uvec$ then $\vvec$.  Note that if
$\lamv = 0$ or $\lamu = 0$, then the coordinate update for $\uvec$ or
$\vvec$ is given by the GMD updates in Step (b) (i) of the GMD
Algorithm.  Consider the following result:
\begin{theorem}
\label{gpmf_sol}
Assume that $P_{1}()$ and $P_{2}()$ are convex and homogeneous of
order one: $P(c x) = c P(x) \ \forall \ c > 0$.  Let $\uvec$ be fixed at
$\uvec'$ or $\vvec$ be fixed at $\vvec'$ such that $\uvec'^{T} 
\Q \X \neq 0$ or $\vvec'^{T} \R \X^{T} \neq 0$. 
Then, the coordinate updates, $\uvec^{*}$ and $\vvec^{*}$, maximizing
the single-factor GPMF problem, \eqref{gpmf}, are given by the
following:  Let $\hat{\vvec} = \argmin_{\vvec} \{
\frac{1}{2}|| \X^{T} \Q \uvec' - \vvec ||_{\R}^{2} + \lamv P_{1}(\vvec)
\}$ and $\hat{\uvec} = \argmin_{\uvec} \{ \frac{1}{2} || \X \R \vvec' -
\uvec ||_{\Q}^{2} + \lamu P_{2}(\uvec) \}$.  Then, 
\begin{align*}
\vvec^{*} = \begin{cases} \hat{\vvec} / || \hat{\vvec} ||_{\R} &
  \textrm{if} \ || \hat{\vvec} ||_{\R} > 0 \\ 0 &
  \textrm{otherwise},  \end{cases} \hspace{4mm} \& \hspace{4mm}
  \uvec^{*} = \begin{cases} \hat{\uvec} / || \hat{\uvec} ||_{\Q} & 
  \textrm{if} \ || \hat{\uvec} ||_{\Q} > 0 \\ 0 &
  \textrm{otherwise}.  \end{cases}
\end{align*}
\end{theorem}

Theorem \ref{gpmf_sol} states that for penalty functions that are
convex and homogeneous of order one, the coordinate updates giving the
single-factor GPMF solution can be obtained by a generalized penalized
regression problem.  Note that penalty functions that are norms or
semi-norms are necessarily convex and homogeneous of order one.  This
class of functions includes common penalties such as the lasso
\citep{tibshirani_lasso}, group 
lasso \citep{yuan_2006_group}, fused 
lasso \citep{tibshirani_2005_fused}, and the generalized lasso
\citep{tibshirani_gen_lasso_2011}.  Importantly, these do not include the ridge
penalty, elastic net, concave penalties such as 
SCAD, and quadratic smoothing penalties commonly used in functional
data analysis.  Many of the penalized regression solutions for
these penalties, however, can be approximated by penalties that are norms or
semi-norms. For instance, to mimic the effect of a given quadratic penalty,
one may use the square-root of this quadratic penalty. Similarly,
the natural majorization-minimization algorithms for SCAD-type
penalties \citep{fan_scad}
involve convex piecewise linear majorizations of the penalties that
satisfy the conditions of Theorem \ref{gpmf_sol}.  Thus, our
single-factor GPMF problem both avoids the complicated scaling problems
of some existing two-way regularized matrix factorizations, and
still permits a wide class of penalties to be used within our
framework.

Following the structure of the GMD Algorithm, the multi-factor GPMF
can be computed via the power method framework.  That is, the GPMF
algorithm replaces Steps (b) (i) of the GMD Algorithm with the single
factor GPMF updates given in Theorem \ref{gpmf_sol}.  This follows the
same approach as that of other two-way matrix factorizations
\citep{huang_2009, witten_pmd_2009, lee_ssvd_2010}.  Notice that
unlike the GMD, subsequent factors computed via this greedy deflation
approach will not be orthogonal in the $\Q,\R$-norm to previously
computed components.  Many have noted in the Sparse PCA literature for
example, that orthogonality of the sparse components is not
necessarily warranted \citep{jolliffe_spca_2003, zou_spca_2006,
  shen_spca_2008}.   
If only one factor is penalized, however,
orthogonality can be enforced in subsequent components via a
Gram-Schmidt scheme \citep{golub_van_loan_1996}.

%here mention - if only one factor is penalized - then can
%orthogonalize the other in the Q,R-norm via Grahm-Schmit
%% We pause briefly to discuss the implications of Theorem \ref{gpmf_sol}
%% in the context of other regularized matrix factorizations discussed in
%% the previous section.  

%% Existing methods
%% either choose to re-scale the factors to have norm one after each
%% iteration or to re-scale the factors after convergence of the algorithm.
%% The latter approach requires complicated checks on the scaling of the
%% factors to avoid numerical instability (\citep{lee_ssvd_2010} and see
%% especially the 
%% supplementary materials of \citet{huang_2009}), but does not put a
%% constraint on the types of penalties that may be employed.  With the
%% former, however, the scaling problems are avoided, but in
%% prior literature, it remained unclear what types of penalties may be
%% used in this context \citep{witten_pmd_2009}.  Thus, Theorem
%% \ref{gpmf_sol} clarifies which classes of penalty functions are
%% appropriate when matrix factors are to be re-scaled at each
%% iteration.  In the context of our GPMF, we chose to re-scale the
%% matrix factors for two reasons.  First, re-scaling the factors yields
%% the same algorithmic framework as that of the GMD (or SVD); and
%% second, if the factors were not re-scaled, then numerical
%% instabilities would be further complicated by the generalizing
%% operators.  

In the following sections, we give
methods for obtaining the single-factor GPMF updates for sparse or
smooth penalty types, noting that many other penalties may also be employed
with our methods.  As the single-factor GPMF problem is symmetric
in $\uvec$ and $\vvec$, we solve the $\R$-norm penalized
regression problem in $\vvec$, noting that the solutions are analogous
for the $\Q$-norm penalized problem in $\uvec$.

\subsection{Sparsity: Lasso and Related Penalties}

%% For many of the penalties described above,
%% there also exist fast algorithms to solve the corresponding
%% penalized regression problem \cite{nesta,yuan_2006_group,tfocs}. Rather
%% than exhaustively describe each specific algorithm, 
%% the following section describes
%% methods for obtaining the single-factor GPMF updates for the LASSO as this is perhaps the most likely version of the GPMF to be used. 
%% We also note that there are many fast methods to solve the LASSO
%% \citep{coor_descent_2007,nesta}, any of which can be used to solve the GPMF.

With high-dimensional data, sparsity in the factors or principal
components yields greater
interpretability and, in some cases have better properties
than un-penalized principal components
\citep{jolliffe_spca_2003, johnstone_spca_2004}.  With
neuroimaging data, sparsity in the factors associated
with the voxels is particularly warranted as in most cases relatively few
 brain regions are expected to truly contribute
to the signal.  Hence, we
consider our GPMF problem, \eqref{gpmf}, with the lasso
\citep{tibshirani_lasso}, or 
$\ell_{1}$-norm penalty, commonly used to encourage sparsity.

The penalized regression problem given in Theorem
\ref{gpmf_sol} can be 
written as a lasso problem: $\frac{1}{2}|| \X^{T} \Q \uvec
- \vvec ||_{\R}^{2} + \lamv || \vvec ||_{1}$.
If $\R = \mathbf{I}$, then the solution for $\hat{\vvec}$ is obtained
by simply applying the soft thresholding operator: $\hat{\vvec} =
S( \X^{T} \Q \uvec, \lambda )$, where $S( x, \lambda) =
\mathrm{sign}(x)( | x| - \lambda)_{+}$ \citep{tibshirani_lasso}.  When
$\R \neq \mathbf{I}$, the solution is not as simple:  
\begin{claim}
\label{gpmf_sparse_diag}
If $\R$ is diagonal with strictly positive diagonal entries, then
$\hat{\vvec}$ minimizing the $\R$-norm lasso problem is given by
$\hat{\vvec} = S( \X^{T} \Q \uvec, \lamv
\R^{-1} \mathbf{1}_{(p)} )$.  
\end{claim}
\begin{claim}
\label{gpmf_sparse_gen}
The solution, $\hat{\vvec}$, that minimizes the $\R$-norm lasso
problem can be obtained by 
iterative coordinate-descent where the solution for each coordinate
$\hat{\vvec}_{j}$ is given by:
$\hat{\vvec}_{j} =  \frac{1}{ \R_{jj} } S
\left( \R_{rj} \X^{T} \Q \uvec  - \R_{j, \neq j} \hat{\vvec}_{\neq j}  ,
\lamv \right)$,
with the subscript  $\R_{rj}$ denoting the row elements associated
with  column $j$ of $\R$.  
\end{claim}

Claim \ref{gpmf_sparse_diag} states that when $\R$ is diagonal, the
solution for 
$\vvec$ can be obtained by soft thresholding the elements of $\y$
by a vector penalty parameter.
For general $\R$, Claim \ref{gpmf_sparse_gen} gives that we can use
coordinate-descent to find the solution without having to compute
$\tilde{\R}$.
Thus, the sparse single-factor GPMF solution can be calculated in a
computationally efficient manner.  One may further improve the speed
of coordinate-descent by employing warm starts and iterating over
active coordinates as described in \citet{friedman_glmnet_2010}.  

We have discussed the details of computing the GPMF factors for the
lasso penalty, and similar techniques can be used to
efficiently compute the solution for the group 
lasso \citep{yuan_2006_group}, fused 
lasso \citep{tibshirani_2005_fused}, generalized lasso
\citep{tibshirani_gen_lasso_2011} and
other sparse convex penalties.  As mentioned, we limit our class of
penalty functions to those that are norms or semi-norms, which does
not include popular concave penalties
such as the SCAD penalty \citep{fan_li_2001}.  
As mentioned above, these penalties, however, can
still be used in our framework as 
concave penalized regression problems can be solved
via iterative weighted lasso problems \citep{mazumder_2009}. 
Thus, our GPMF framework can be used with a wide
range of penalties to obtain sparsity in the factors.  

Finally, we note that as our GMD Algorithm can be used to perform GPCA,
the sparse GPMF framework can be used to find sparse generalized principal
components by setting $\lamu = 0$ in 
\eqref{gpmf}. This is an approach well-studied in the Sparse PCA
literature \citep{shen_spca_2008, witten_pmd_2009, lee_ssvd_2010}.

\subsubsection{Smoothness: $\Omeg$-norm Penalty \& Functional Data}

In addition to sparseness, there is much interest in penalties that
encourage smoothness, especially in the context of functional data
analysis.  We show how the GPMF can be used with smooth
penalties and propose a generalized gradient descent method to solve for these
smooth GPMF factors.  
Many have proposed to obtain smoothness in the factors by using a
quadratic penalty.  \citet{rice_silverman_1991} suggested $P(\vvec) =
\vvec^{T} \Omeg \vvec$, where $\Omeg$ is the matrix of squared second or
fourth differences. As this penalty is not homogeneous of order one,
we use the $\Omeg$-norm penalty: $P(\vvec) = ( \vvec^{T}
\Omeg \vvec )^{-1/2} = || \vvec ||_{\Omeg}$.  Since this  penalty is a
norm or a semi-norm, the GPMF
solution given in Theorem \ref{gpmf_sol} can be employed.  We
seek to minimize the following $\Omeg$-norm penalized regression problem:
$\frac{1}{2} ||  \X^{T} \Q
\uvec -  \vvec ||_{\R}^{2} + \lamv || \vvec ||_{\Omeg}$.
Notice that this problem is similar in structure to the group lasso problem of
\citet{yuan_2006_group} with one group.

To solve the $\Omeg$-norm penalized regression problem, we use a
generalized gradient descent method. 
(We note that there are more elaborate
first order solvers, introduced in recent works such as \citet{tfocs, nesta},
and we describe a simple version of such solvers).
Suppose $\Omeg$ has rank $k$. Then, set
$\y = \X^{T} \Q \uvec$, and define $\Bmat$ as
$\Bmat =
\begin{pmatrix}
\Omeg^{-1/2} \\
 \mathbf{N}
\end{pmatrix}
$ where $\Omeg^{-1/2} \in \Re^{k \times p}$ and $\mathbf{N} \in
\Re^{(p-k) \times p}$.  The rows of $\mathbf{N}$ are taken to span the null space of $\Omeg$,
and $\Omeg^{-1/2}$ is taken to satisfy $(\Omeg^{-1/2})^{T} \Omeg
\Omeg^{-1/2} = \mathbf{P}_{\Omeg}$,
the Euclidean projection onto the column space of $\Omeg$.
The matrices $\Omeg^{-1/2}$ and $\mathbf{N}$ can be found from the full
eigenvalue decomposition of $\Omeg$, $\Omeg = \Gam \Lam^{2} \Gam^{T}$. 
Assuming $\Lam$ is in decreasing order,
we can take $\Omeg^{-1/2}$ to be 
$\Lam^{-1}(1:k,1:k) (\Gam (1:k,\cdot))^{T}$ and $\mathbf{N}$ to be the
last $p-k$ rows of $\Gam$. 
If $\Omeg$ is taken to denote the squared differences, for example, then
$\mathbf{N}$ can be taken as a row of ones.  For the squared 
second differences, $\mathbf{N}$ can be set to
have two rows: a row of ones and a row with the linear sequence $1,
\ldots, p$.

Having found $\Omeg^{-1/2}$ and $\mathbf{N}$, 
we re-parametrize the problem by taking a non-degenerate linear
transformation of $\vvec$ by setting
$\Bmat^{-1} \vvec
= 
\begin{pmatrix}
  \wvec \\
  \mathbf{\eta}
\end{pmatrix}
$ so that $\Bmat \wvec
= \vvec$.  Taking $\Omeg^{1/2}$ to be $\Lam(1:k,1:k) (\Gam (1:k,\cdot))^{T}$, we note that 
$\|\vvec\|_{\Omeg} = \|\Omeg^{1/2}\vvec\|_2 =
\|\Omeg^{1/2}(\Omeg^{-1/2}\wvec + \mathbf{N}\eta)\|_2 = \|\wvec\|_2$,
as desired.
The $\Omeg$-norm penalized regression problem, written
in terms of $(\wvec, \mathbf{\eta})$ therefore has the form
\begin{align}
\label{pen_reg_grp_repar}
\frac{1}{2} || \y - \Omeg^{-1/2} \wvec - \mathbf{N} \mathbf{\eta} ||_{\R}^{2} +
\lamv || \wvec ||_{2}.
\end{align}
The solutions of \eqref{pen_reg_grp_repar} are in one-to-one
correspondence to those of the $\Omeg$-norm regression problem via the
relation $\Bmat\wvec = \vvec$ and hence, it is sufficient 
to solve \eqref{pen_reg_grp_repar}.  
Consider the following algorithm and result:
\begin{algorithm}
\caption{Algorithm for re-parametrized $\Omeg$-norm penalized regression.}
\label{alg_smooth}
\begin{enumerate}
\item Let $L > 0$ be such that $|| \Bmat^{T} \R \Bmat ||_{op} < L$ and
  initialize $\wvec^{(0)}$. 
\item Define $\tilde{\wvec}^{(t)} = \wvec^{(t)} + 
  \frac{1}{L} \Bmat^{T} \R \left( \y - \Omeg^{-1/2} \wvec^{(t)} -
  \mathbf{N} \mathbf{\eta}^{(t)}\right)$.  \\
Set 
$\wvec^{(t+1)} = \left( 1 - \frac{\lamv}{L ||  \tilde{\wvec}^{(t)} 
  ||_{2}}\right)_{+}   \tilde{\wvec}^{(t)}$.
\item Set
$\mathbf{\eta}^{(t+1)} = (\mathbf{N}^T\R\mathbf{N})^{\dagger} \left(\mathbf{N}^T\R (\y - \Omeg^{-1/2} \wvec^{(t+1)}) \right)$.
\item Repeat Steps (b)-(c) until convergence.
\end{enumerate}
\end{algorithm}

\begin{proposition}
\label{prop_gpmf_smooth}
The $\vvec^{*}$ minimizing the $\Omeg$-norm penalized
regression problem is given by $\vvec^{*}
= \Bmat
\begin{pmatrix}
\wvec^{*} \\
\mathbf{\eta}^*  
\end{pmatrix}
$ where $\wvec^{*}$  and $\eta^{*}$ are the solutions of Algorithm
~\ref{alg_smooth}. 
\end{proposition}

Since our problem 
employs a monotonic increasing function of the quadratic smoothing
penalty, $|| \vvec ||_{\Omeg}^{2}$,
typically used for functional data analysis, then, the $\Omeg$-norm
penalty also results in a smoothed factor, $\vvec$.

Computationally, our algorithm requires
taking the matrix square root of the penalty matrix $\Omeg$.  For 
dense matrices, this is of order $O(p^{3})$, but for commonly used
difference matrices, sparsity can reduce the
order to $O(p^{2})$ \citep{golub_van_loan_1996}.  The matrix
$\Bmat$, can be computed and stored
prior to running algorithm and $\tilde{\R}$ does not need to be
computed.  Also,  each step of Algorithm ~\ref{alg_smooth} is on the order
of matrix multiplication.  The generalized gradient descent steps
for solving for $\wvec^{(t+1)}$ can also be solved
by Euclidean projection onto 
$\left\{\vvec \in \Re^p: \vvec'\Omeg \vvec \leq \lambda \right\}$.
Hence, as long as this projection is computationally feasible,
the smooth GPMF penalty is
computationally feasible for high-dimensional
functional data. 

We have shown how one can use penalties to
incorporate smoothness into the factors of our GPMF model.  As with
sparse GPCA, we can use this smooth penalty to perform functional
GPCA.  The analog of our $\Omeg$-norm penalized single-factor GPMF
criterion with $\lamu = 
0$ is closely related to previously proposed methods for computing functional
principal components \citep{huang_fpca_2008, huang_2009}.

\subsection{Selecting Penalty Parameters \& Variance Explained}

When applying the GPMF and Sparse or Functional GPCA to real structured
data, there are two practical considerations that must be addressed:
(1) the number of factors, $K$, to extract, and (2) the choice of
penalty parameters, $\lamu$ and $\lamv$, controlling the amount of
regularization.  Careful examination of the former is beyond the
scope of this paper.  For classical PCA, however, several methods
exist for selecting the number of factors \citep{buja_1992, troyanskaya_2001,
  owen_2009_cv}.  Extending the 
imputation-based methods of \citet{troyanskaya_2001} for GPCA methods is
straightforward; we believe extensions of \citet{buja_1992} and
\citet{owen_2009_cv} are also possible in our framework.  A related issue to
selecting the number of factors is the amount of variance explained.
While this is a simple calculation for GPCA (see Corollary
~\ref{cor_gpca_var_ex}), this is not as straightforward for two-way
regularized GPCA as the factors are no longer orthonormal in the
$\Q,\R$-norm.  Thus one must project out the effect of the previous
factors to compute the cumulative variance explained by the first
several factors.
\begin{proposition}
\label{prop_var_ex}
Let $\U_{k} = [ \uvec_{1} \ \ldots \ \uvec_{k} ]$ and $\V_{k} = [
  \vvec_{1} \ \ldots \ \vvec_{k} ]$ and define $\Pmat_{k}^{(\U)} =
  \U_{k} ( \U_{k}^{T} \Q \U_{k} )^{-1} \U_{k}^{T}$ and  $\Pmat_{k}^{(\V)} =
  \V_{k} ( \V_{k}^{T} \R \V_{k} )^{-1} \V_{k}^{T}$ yielding $\X_{k} =
  \Pmat_{k}^{(\U)} \Q \X \R \Pmat_{k}^{(\V)}$.  Then, the cumulative
  proportion of variance explained by the first $k$ regularized GPCs
  is given by: $\mathrm{tr} ( \Q \X_{k} \R \X_{k}^{T} ) / \mathrm{tr}
  ( \Q \X \R \X^{T} )$.  
\end{proposition}
Also note that since the deflation-based GPMF algorithm is greedy, the
components are not necessarily ordered in terms of variance explained
as those of the GMD.

%things to mention here
%relevant range of lamu and lamv
%starting points for GPMF
%better explain exactly how we implement the BIC method

When applying the GPMF, the penalty parameters $\lamu$ and $\lamv$
control the amount of sparsity or smoothness in the estimated
factors.  We seek a data-driven way to estimate these penalty
parameters.   Many 
have proposed cross-validation approaches for the SVD \citep{troyanskaya_2001,
  owen_2009_cv} or nested generalized cross-validation or Bayesian Information
Criterion (BIC) 
approaches \citep{huang_2009, lee_ssvd_2010}.  While all of these
methods are appropriate for our models as well, we propose an
extension of the BIC
selection method of \citet{lee_ssvd_2010} appropriate for the
$\Q,\R$-norm.   
\begin{claim}
\label{claim_bic}
The BIC selection criterion for the GPMF factor $\vvec$ with $\uvec$
and $d$ fixed at $\uvec'$ and $d'$ respectively is given by the
following:
\begin{align*}
 BIC(\lamv) &=  \ \mathrm{log} \left(
\frac{|| \X - d' \uvec' \hat{\vvec} ||_{\Q,\R}^{2}}{np} \right)  +
\frac{\mathrm{log} (np)}{np}  \widehat{df}(\lamv). %% \\
%% \textrm{where} \ \ \widehat{df}(\lamv) = \mid \{ \hat{\vvec} \}
%% \mid  \ &  \textrm{(Lasso Penalty)}.%  \ \ \& \ \  \widehat{df}(\lamv) = (p-1) \frac{|| \hat{\vvec} ||_{\R}}{|| \X^{T} \Q
%% % \uvec' ||_{\R}}  \  \textrm{($\Omeg$-norm Penalty)}.
\end{align*}
\end{claim}
The BIC criterion for the other factor $\uvec$ is analogous.  
Here, $\widehat{df}(\lamv)$ is an estimate of the degrees of freedom
for the particular penalty employed.  For the lasso penalty, for
example, $\widehat{df}(\lamv) = \mid \{ \hat{\vvec} \} \mid$
\citep{zou_dof_2007}.  Expressions for the degrees of freedom of other penalty
functions are given in \citet{kato_dof_2009} and
\citet{tibshirani_gen_lasso_2011}.

Given this model selection criterion, 
one can select the optimal
penalty parameter at each iteration of the factors for the GPMF
algorithm as in \citet{lee_ssvd_2010}, or use the criterion to select
parameters after the iterative algorithm has converged.  In 
our experiments, we found that both of these approaches perform
similarly, but the latter is more numerically stable.  The performance
of this method is investigated through simulation studies in Section
~\ref{section_simulations}.  Finally, we note that since the the GPMF only
converges to a local optimum, the solution achieved is highly
dependent on the initial starting point.  Thus, we use the GMD solution
to initialize all our algorithms, an approach similar to related
methods which only achieve a local optimum \citep{witten_pmd_2009,
  lee_ssvd_2010}.

\section{Results}
\label{section_results}

We assess the effectiveness of our methods on simulated data sets and
a real functional MRI example.

\subsection{Simulations}
\label{section_simulations}

 All simulations are generated from the following model: $\X = \Smat +
\Emat = \sum_{k=1}^{K} \phi_{k} \uvec_{k} \vvec_{k}^{T} + \Sig^{1/2}
\Zmat \Delt^{1/2}$, where the
$\phi_{k}$'s denote the magnitude of the rank-$K$ signal matrix $\Smat$,
$\Zmat_{ij} \overset{iid}{\sim} N(0,1)$ and $\Sig$ and $\Delt$ are the row
and column 
covariances so that $\Emat \sim N_{n,p}(0,\Sig,\Delt)$.  Thus, the
data is simulated according to the matrix-variate normal 
distribution with mean given by the rank-$K$ signal matrix, $\Smat$.
The SNR for this model is given by 
$\mathrm{E}\left[ \mathrm{tr}( \Smat^{T} \Smat )\right] /
\mathrm{E}\left[ \mathrm{tr} ( \Emat^{T} \Emat ) \right] =
\mathrm{E}\left[ \sum_{k=1}^{K} \phi_{k}^{2} \uvec_{k}^{T} \uvec_{k}
  \vvec_{k}^{T} \vvec_{k}  \right] /
\mathrm{E} \left[ \mathrm{tr}(\Sig) \mathrm{tr}(\Delt) \right]$
\citep{gupta_nagar}.  The data
is row and column centered before each method is applied, and to be
consistent, the BIC criterion is used to select parameters for all methods.

%% \begin{figure}
%% \begin{center}
%%  \includegraphics[width=5.5in]{fig_mexample1_parta.jpg}
%%  \includegraphics[width=4.5in]{fig_mexample1_partb.jpg}
%% \end{center}
%% \caption{\em \footnotesize Spatio-Temporal Simulation Example: A rank
%%   one signal 
%%   comprised of the outer product of a random spatial factor with three
%%   regions of interest (top panel, top left) and a temporal factor given by a sinusoidal
%%   activation pattern (bottom panel, in black) is confounded by adding spatio-temporal noise
%%   following an autoregressive Gaussian Markov random field.}
%% \label{fig_mexample1}
%% \end{figure}

\subsubsection{Spatio-Temporal Simulation}

\begin{figure}[!!t]
\includegraphics[width=5in,clip=true,trim=0in 1in 0in 0in]{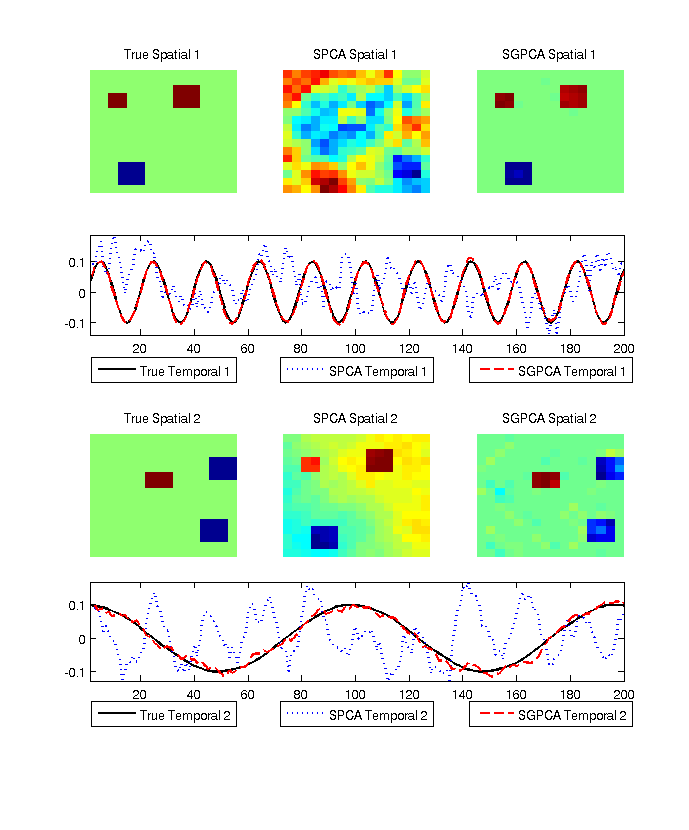}
\caption{\em Example results from the spatio-temporal simulation.}
\label{fig_spt_ex}
\end{figure}

Our first set of simulations is inspired by spatio-temporal
neuroimaging data, and is that of the example shown in Figure
~\ref{fig_spt_ex}.   The two spatial signals, $\U \in \Re^{256
  \times 2}$, each consist of three non-overlapping ``regions of
interest'' on a $16 \times 16$ grid.  The two temporal signals, $\V
\in \Re^{200 \times 2}$, are sinusoidal activation patterns with
different frequencies.  The signal is given by $\phi_{1} \sim N(1,\sigma^{2})$
and $\phi_{2} \sim N(0.5, \sigma^{2})$ where $\sigma$ is
chosen so that the SNR = $\sigma^{2}$.  The noise is simulated from an
autoregressive Gaussian Markov random field.  The spatial covariance,
$\Sig$, is that of an autoregressive process on the $16 \times 16$
grid with correlation $0.9$, and the temporal covariance, $\Delt$, is
that of an autoregressive temporal process with correlation $0.8$.

\begin{table}
\caption{\em Comparison of different quadratic operators for GPCA in the
  spatio-temporal simulation.}
\hspace{-.65in} \begin{tabular}{l|r|r|r|r|r|r}
& \% Var $k=1$ & \% Var $k=2$ & MSSE $\uvec_{1}$ & MSSE
  $\uvec_{2}$ & MSSE $\vvec_{1}$ & MSSE $\vvec_{2}$ \\
\hline
$\Q = \I_{(m^2)}$, $\R = \I_{(p)}$ &  57.4 (0.9)  &  19.1 (0.9) &
  0.5292 (.04) &    0.6478 (.02) &    0.3923 (.02) &    0.4857 (.01) \\
$\Q = \Sigi$, $\R = \Delti$ &  75.4 (2.3)  &  6.6 (0.2) & 
  0.1452 (.02) &    0.3226 (.02) &    0.0087 (.01) &    0.0180 (.01) \\
$\Q = \Lmat_{m,m}$, $\R = \Lmat_{p}$ &    42.9 (2.9) &  2.8 (0.1) &
  0.1981 (.03) &    0.7972 (.02) &   0.0609 (.02) &  0.4334  (.02) \\ 
$\Q = \Lmat_{m,m}$, $\R = \Smat_{p}$ &   55.8 (2.3) &  14.5 (0.3)  &
  0.1714 (.02) &    0.3425 (.03) &    0.0481 (.01) &    0.0809 (.02) \\
$\Q = \Smat_{m,m}$, $\R = \Lmat_{p}$ &   67.6 (0.7) &    13.0 (0.8) &
  0.8320 (.02) &    0.8004 (.01) &    0.5414 (.01) &  0.4831 (.01) \\ 
$\Q = \Smat_{m,m}$, $\R = \Smat_{p}$ &  60.3 (1.0) &   19.3 (1.0)  &
  0.5682 (.04) &    0.6779 (.02) &   0.4310 (.02) &    0.5030 (.01) \\
\end{tabular}
\label{tab_qr}
\end{table}

The behavior demonstrated by Sparse PCA (implemented using the
approach of \citep{shen_spca_2008}) and Sparse GPCA in Figure
~\ref{fig_spt_ex} where $\sigma^{2} = 1$ is typical.  Namely, when the
signal is small and the structural noise is strong, PCA and Sparse PCA
often find structural noise as the major source of variation instead
of the true signal, the regions of interest and activation patterns.
Even in subsequent components, PCA and Sparse PCA often find a mixture
of signal and structural noise, so that the signal is never easily
identified.  
In this example, $\Q$ and $\R$ are not estimated from the data and
instead fixed based on the known data 
structure, a $16 \times 16$ grid and $200$ equally spaced points
respectively.  In Table ~\ref{tab_qr}, we explore two simple
possibilities for quadratic operators for this spatio-temporal data:
graph Laplacians of a graph connecting nearest neighbors and a kernel
smoothing matrix (using the Epanechnikov kernel) smoothing nearest
neighbors.  Notationally, these are denoted as $\Lmat_{m,m}$ for a
Laplacian on a $m \times m$ mesh grid or $\Lmat_{p}$ for $p$ equally
spaced variables; $\Smat$ is denoted analogously.  
We present results when $\sigma^{2} = 1$ in terms of
variance explained and mean 
squared standardized error (MSSE) for 100 replicates of the
simulation.  Notice that the combination of a 
spatial Laplacian and a temporal smoother in terms of signal recovery
performs nearly as well as 
when $\Q$ and $\R$ are set to the true population values and
significantly better than classical PCA.  Thus, quadratic
operators do not necessarily need to be estimated from the data, but
can instead 
be fixed based upon known structure.  Returning to the example in
Figure ~\ref{fig_spt_ex}, $\Q$ and $\R$ are taken as a Laplacian and
smoother respectively, yielding excellent recovery of the regions of
interest and activation patterns.

\begin{figure}
\hspace{-.25in}\includegraphics[width=6in,clip=true,trim=1.1in 0in 1in 0in]{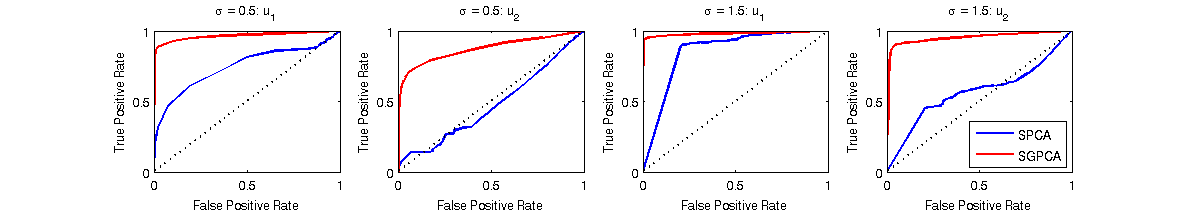}
\caption{\em Average receiver-operator curves for the spatio-temporal
  simulation.}
\label{fig_roc}
\end{figure}

Next, we investigate the feature selection properties of Sparse GPCA
as compared to Sparse PCA for the structured spatio-temporal
simulation.  Note that given the above results, $\Q$ and $\R$ are
taken as a Laplacian and a smoothing operator for nearest neighbor
graphs respectively.  
In Figure ~\ref{fig_roc}, we present mean receiver-operator
curves (ROC) averaged over 100 replicates achieved by varying the
regularization parameter, $\lambda$.  In Table ~\ref{tab_select}, we
present feature selection results in terms of true and false positives
when the regularization parameter $\lambda$ is fixed and estimated via
the BIC method.  From both of these results, we see that Sparse GPCA
has major advantages over Sparse PCA.  Notice also that accounting for
the spatio-temporal structure through $\Q$ and $\R$ also yields better
model selection via BIC in Table \ref{tab_select}.

Overall, this spatio-temporal simulation has demonstrated significant
advantages 
of GPCA and Sparse GPCA for data with low signal compared to the
structural noise when the structure is fixed and known.

\begin{table}
\caption{\em Feature selection results for the spatio-temporal
  simulation.}
\begin{tabular}{ll|r|r|r|r}
& & True Positive $\uvec_{1}$ & False Positive $\uvec_{1}$ & True
  Positive $\uvec_{2}$ & False Positive $\uvec_{2}$ \\
\hline
\multirow{2}{*}{$\sigma = 0.5$} & Sparse PCA &   0.9859 (.008) &    0.9758 (.014) &   0.7083 (.035)  &
0.7391 (.029) \\
& Sparse GPCA  &  0.9045 (.025) &   0.0537 (.007) &   0.7892 (.032) &
0.1749 (.005) \\
\hline
\multirow{2}{*}{$\sigma = 1.5$} & Sparse PCA &    0.9927 (.004) &   0.7144 (.039)  &  0.7592 (.034) &
0.7852 (.030) \\
& Sparse GPCA  &  0.9541 (.018) &   0.0292 (.005)  &  0.9204 (.024)  &
0.1572 (.004) \\
\end{tabular}
\label{tab_select}
\end{table}

\subsubsection{Sparse and Functional Simulations}

The previous simulation example used the outer product of a  sparse
and smooth signal with 
autoregressive noise.  Here, we investigate the behavior of our
methods when both of the factors are either smooth or sparse and when the
noise arises from differing graph structures.  Specifically, we simulate
noise with 
either a block diagonal covariance or a random graph structure in
which the precision matrix arises from a graph with a random number of
edges. The former has blocks of size five with off diagonal elements
equal to $0.99$.  The latter is generated from a graph where each
vertex is connected to $n_{j}$ randomly selected
vertices, where $n_{j} \overset{iid}{\sim} Poisson(3)$ for the row
variables and 
$Poisson(1)$ for the column variables.  The row and column
covariances, $\Sig$ and $\Delt$ are taken as the inverse of the nearest
diagonally dominant matrix to the graph Laplacian of the row and
column variables respectively.  For our GMD
and GPMF methods, the values of the true quadratic operators are
assumed to be unknown, but the structure is known, and thus $\Q$ and
$\R$ are taken to be the graph Laplacian.  One hundred
replicates of the simulations 
for data of dimension $100 \times 100$ were performed.

\begin{table}
\caption{\em \footnotesize Functional Data Simulation Results. %% : Functional
%%   data of dimension $100 \times 100$ with
%%   signal given by the outer product of two sinusodal curves and noise
%%   with either a block diagonal covariance structure or from a random
%%   graph structure was simulated $100$ times with signal-to-noise ratio
%%   given by $\sigma^{2}$.  The average squared errors of the row
%%   factor, column factor, and rank one signal are compared for the SVD
%%   (PCA), two-way Functional PCA, the GMD (GPCA) and the Functional
%%   GPMF.  For the latter two,
%%   generalizing operators 
%%   were assumed to be unkown and taken to be the unweighted grahp
%%   Laplacians. 
}
\begin{tabular}{ll|rrr}
&&  Squared Error & Squared Error & Squared Error  \\
&&  Row Factor & Column Factor &  Rank 1 Matrix \\
\hline
\hline
$\sigma = 0.5$  &&&& \\
&Block Diagonal Covariance &&& \\
\hline
&SVD &   3.670 (.121)  &  3.615 (.117) &  67.422 (3.29)   \\
&Two-Way Functional PCA  & 3.431 (.134)  &  3.315 (.127) &  65.973
(3.41)  \\
&GMD &   1.779 (.138)  &  2.497 (.138) &  60.687  (3.55)   \\
&Functional GPMF  &  1.721 (.133)  &  1.721  (.144) &   56.296 (2.81)  \\
\hline
&Random Graph &&& \\
\hline
&SVD &   2.827 (.241) &   2.729 (.234)  & 49.674 (1.50)   \\
&Two-Way Functional PCA &   1.968 (.258)  &  1.878 (.250) &  49.694
 (1.50)  \\
&GMD &    0.472 (.095)  &  1.310 (.138)  & 49.654 (1.51)   \\
&Functional GPMF &   0.663 (.049)  &  0.388 (.031) &  49.666 (1.51)  \\
\hline
\hline
$\sigma = 1.0$ &&&& \\
&Block Diagonal Covariance &&& \\
\hline
&SVD &     2.654
(.138) &  2.532 (.143) &  94.412 (6.49) \\
&Two-Way Functional PCA   &   2.426 (.142)  &  2.311 (.144) &  92.522 (6.59) \\
&GMD  &   1.094
(.117) &   1.647 (.120)  & 88.852 (6.73) \\
&Functional GPMF   &
1.671 (.110)  &  2.105 (.129)  & 76.864 (5.27) \\
\hline
&Random Graph &&& \\
\hline
&SVD &    2.075
 (.224) &   1.961 (.226) &  50.392 (2.61) \\
&Two-Way Functional PCA &      1.224 (.212) &   1.187 (.206) &  50.276
(2.62) \\ 
&GMD &  0.338
 (.075)  &  0.926 (.119) &  50.258 (2.62) \\
&Functional GPMF &   
 0.608 (.070) &    0.659 (.245)  & 50.345 (2.60) \\
\hline
\end{tabular}
\label{tab_sim2}
\end{table}

\begin{figure}
\hspace{-.75in}\includegraphics[width=6.75in,clip=true,trim=0in .5in 0in .35in]{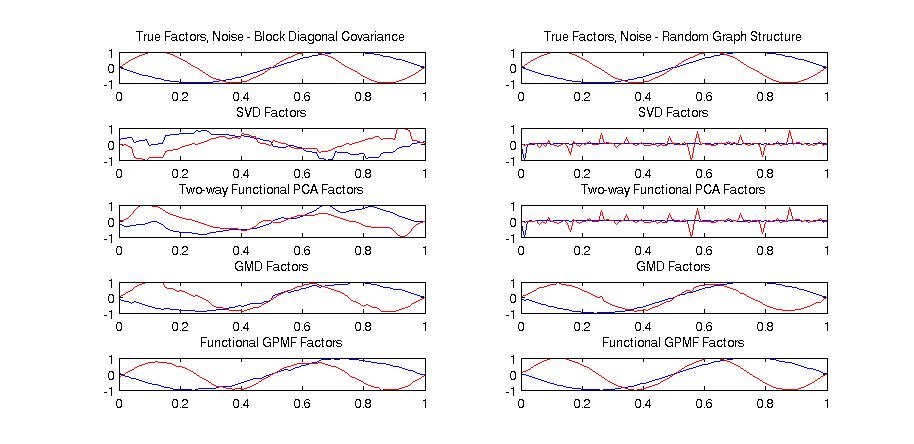}
\caption{\em \footnotesize Example Factor Curves for Functional Data
  Simulation.}
\label{fig_mexample2}
\end{figure}

To test our GPMF method with the smooth $\Omeg$-norm penalties, we
simulate sinusoidal row and column factors: $\uvec = sin (4 \pi x)$ and $\vvec
= -sin(2 \pi x)$ for 100 equally spaced values, $x \in [0,1]$.  As the
factors are fixed, the rank one signal is multiplied by $\phi \sim N(0,
c^{2} \sigma^{2})$, where $c$ is chosen such that the $SNR = \sigma^{2}$.
In Table 3, we compare the GMD and GPMF with smooth penalties to
the SVD and the two-way functional PCA of \citet{huang_2009} in terms
of squared prediction errors of the row and column factors and rank
one signal.  We see that both the GMD and functional GPMF outperform
competing methods.  Notice, however, that the functional GPMF
does not always perform as well as the un-penalized GMD.  In many
cases, as the quadratic operators act as smoothing matrices of the
noise, the GMD yields fairly smooth estimates of the factors, Figure
~\ref{fig_mexample2}.

%%  This likely occurs as the $\Omeg$-norm
%% penalty in addition to the generalizing operator may have a tendency
%% to over-smooth the factors when the noise and signal have a similar
%% structure, as dicussed in Section ~\ref{section_connect}, because the
%% generalizing operators behave as smoothers of the noise.  Thus, if
%% the smooth signal and noise have the same structure, as with the
%% autorgressive noise of the temporal componenet in the first
%% simualtion, then smoothing the factors with the $\Omeg$-norm is
%% unneccessary.  

\begin{table}
\caption{\em \footnotesize Sparse Factors Simulation Results.%% : Sparse data of
%%   dimension $100 \times 100$ with
%%   signal given by the outer product of two vectors with 25 non-zero
%%   elements $\sim^{iid} N(0,\sigma^{2})$ and noise
%%   with either a block diagonal covariance structure or from a random
%%   graph structure was simulated $100$ times.  The average percent true
%%   and false positives of the row and column factors are compared for
%%   Sparse PCA (PMD) and Sparse GPCA (GPMF).  For the latter,
%%   generalizing operators 
%%   were assumed to be unkown and taken to be the unweighted grahp
%%   Laplacians. 
}
\begin{tabular}{l|rrrr}
& \% True Positives & \% False Positives & \% True
  Positives & \% False Positives \\
& Row Factor & Row Factor & Column Factor & Column Factor \\
\hline
\hline
Block Diagonal &&&& \\
Covariance &&&& \\
\hline
\hline
$\sigma = 0.5$ &&&& \\
 Sparse PMD &    76.24 (1.86)  &
  46.20 (3.16) &  79.68 (1.55) &   53.23 (3.16) \\
 Sparse GPMF &   79.72 (2.45)  &  1.16  (0.15) &   82.60 (2.57) &
  1.16 (0.17) \\
\hline
$\sigma = 1.0$ &&&& \\
   Sparse PMD & 87.80 (1.26) &  32.80
  (3.34) &  88.56 (1.05) &   40.93 (3.29) \\
 Sparse GPMF &   87.12 (2.19) & 1.25 (0.17) &  87.24 (2.21) &   1.48
  (0.18) \\
\hline
\hline
Random Graph &&&& \\
\hline
\hline
$\sigma = 0.5$ &&&& \\
 Sparse PMD &    83.56 (1.86) &
  39.01 (3.16) &  79.44 (1.55) &   33.67 (3.16) \\
 Sparse GPMF &   87.36 (2.45) &  28.29 (0.15) &  81.16 (2.53) &
  24.73 (0.17) \\
\hline
$\sigma = 1.0$ &&&& \\
 Sparse PMD &   88.04 (1.26) &
  46.12 (3.34) &  85.36 (1.05) &  48.63 (3.29) \\
 Sparse GPMF &   92.48 (2.19) &   44.28 (0.17) &   87.80 (2.21) &
  41.20 (0.18) \\
\hline
\end{tabular}
\label{tab_sim3}
\end{table}

Finally, we test our sparse GPMF method against other sparse penalized
matrix decompositions \citep{witten_pmd_2009, lee_ssvd_2010}, in Table
4.  In both 
the row and column factors, one fourth of the elements are non-zero
and simulated according to $N(0, \sigma^{2})$.  Here the scaling
factor $\phi$ is chosen so that the $SNR = \sigma^{2}$.  We compare
the methods in terms of the average percent true and false positives
for the row and column factors.  The results indicate that our methods
perform well, especially when the noise has a block diagonal
covariance structure.

The three simulations indicate that our GPCA and sparse and functional
GPCA methods perform well when there are two-way dependencies of the
data with known structure.  For the tasks
of identifying regions of interest, functional patterns, and feature
selection with transposable data, our methods show a substantial
improvement over existing technologies.

\subsection{Functional MRI Example}
\label{section_example}

\begin{figure}
\begin{center}\includegraphics[width=3.5in]{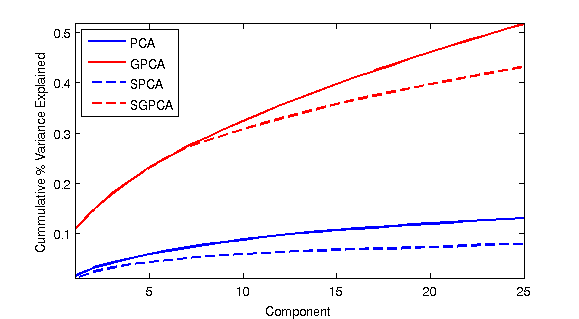}
\caption{\em \footnotesize Cumulative proportion of variance explained
  by the first 25 PCs on the starplus fMRI data.  Generalized PCA
  (GPCA) and Sparse GPCA (SGPCA) yield a significantly larger reduction in
  dimension that PCA and Sparse PCA (SPCA). }
\end{center}
\label{fig_fmri_varex}
\end{figure}

We demonstrate the effectiveness of our GPCA and Sparse GPCA methods
on a functional MRI example.  As discussed in the motivation for our
methods, functional MRI data is a classic example of two-way
structured data in which the nature of the noise with respect to this
structure is relatively well understood.  Noise in the spatial domain
is related 
to the distance between voxels while noise in the temporal domain is
often assumed to follow a autoregressive process or another classical
time series process \citep{lindquist_2008}.  Thus, when fitting our
methods to this 
data, we consider fixed quadratic operators related to these
structures and select the pair of quadratic operators yielding the
largest proportion of variance explained by the first GPC.  Specifically, we
consider $\Q$ in the spatial domain to be a  graph
Laplacian of a nearest neighbor graph connecting the voxels or a
positive power
of this Laplacian.  In the temporal domain, we consider $\R$ as a
graph Laplacian or a positive power of a Laplacian of a chain graph or a
one-dimensional smoothing matrix with a window size of five or ten.
In this manner, $\Q$ and $\R$ are not estimated from the data and are
fixed a priori based on the known two-way structure of fMRI data.

For our functional MRI example, we use a well-studied, publicly
available fMRI data set 
where subjects were shown images and read audible sentences related to
these images, a data set which refer to as the ``StarPlus'' data
\citep{mitchell_fmri_2004}.  We study data for one subject, subject number 04847, which
consists of 4,698 voxels ($64 \times 64 \times 8$ images with
masking) measured for 20 tasks, each lasting for 27 seconds, yielding
54 - 55 time points.  The data was pre-processed by standardizing each voxel
within each task segment.   Rearranging the data yields a $4,698 \times
1,098$ matrix to which we applied our dimension reduction techniques.
For this data, $\Q$ was selected to be an unweighted Laplacian and
$\R$ a kernel smoother with a window size of ten time points.  In
Figure \ref{fig_fmri}, we present the first three PCs found by PCA,
Sparse PCA, GPCA, and Sparse GPCA.  Both the spatial PCs, illustrated
by the eight axial slices, and the corresponding temporal PCs, with
dotted red vertical lines denoting the onset of each new task period,
are given.   Spatial noise overwhelms PCA and Sparse PCA with large
regions selected and with the corresponding time series appearing to
be artifacts or scanner noise, unrelated to the experimental task.  The
time series 
of the first three GPCs and Sparse GPCs, however, exhibit a clear
correlation with the task period and temporal
structure characteristic of the BOLD signal.  While the spatial PCs of
GPCA are 
a bit noisy, incorporating sparsity with Sparse GPCA yields spatial
maps consistent with the task in which the subject was viewing and
identifying images then hearing a sentence related to that
image.  In particular, the first two SGPCs show bilateral occipital,
left-lateralized inferior temporal, and inferior frontal activity
characteristic of the well-known "ventral stream" pathway recruited
during object identification tasks \citep{pennick_2012}.

In Figure \ref{fig_fmri_varex}, we show the
cumulative proportion of variance explained by the first 25 PCs for
all our methods.  Generalized PCA methods clearly exhibit an enormous
advantage in terms of dimension reduction with the first GPC and SGPC
explaining more sample variance than the first 20 classical PCs.  As
PCA is commonly used as an initial dimension reduction technique for
other pattern recognition techniques, such as independent components
analysis, applied to fMRI 
data \citep{lindquist_2008}, our Generalized PCA methods can offer a
major advantage in this 
context.  Overall, by directly accounting for the known two-way
structure of fMRI data, our GPCA methods yield biologically and
experimentally relevant results that explain much more variance than
classical PCA methods.

\begin{sidewaysfigure}
\includegraphics[width=4.25in]{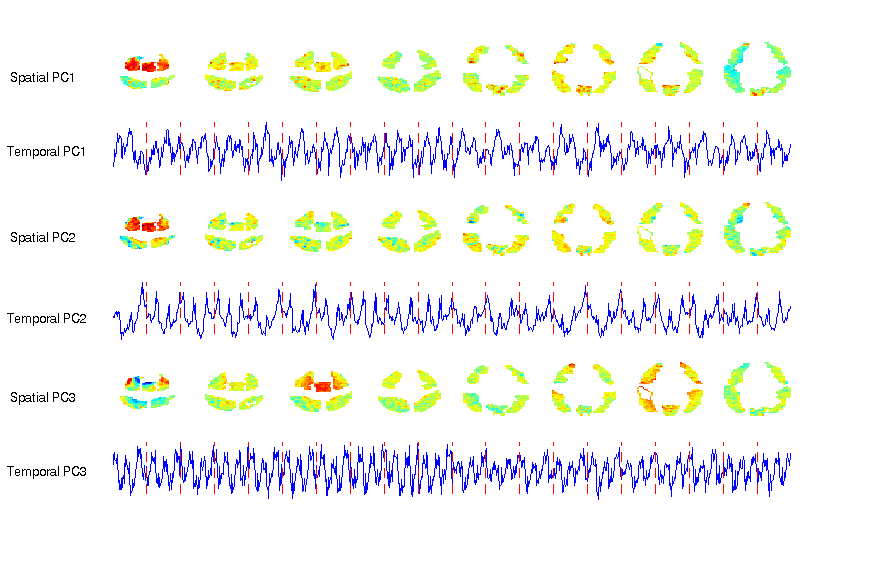}\includegraphics[width=4in]{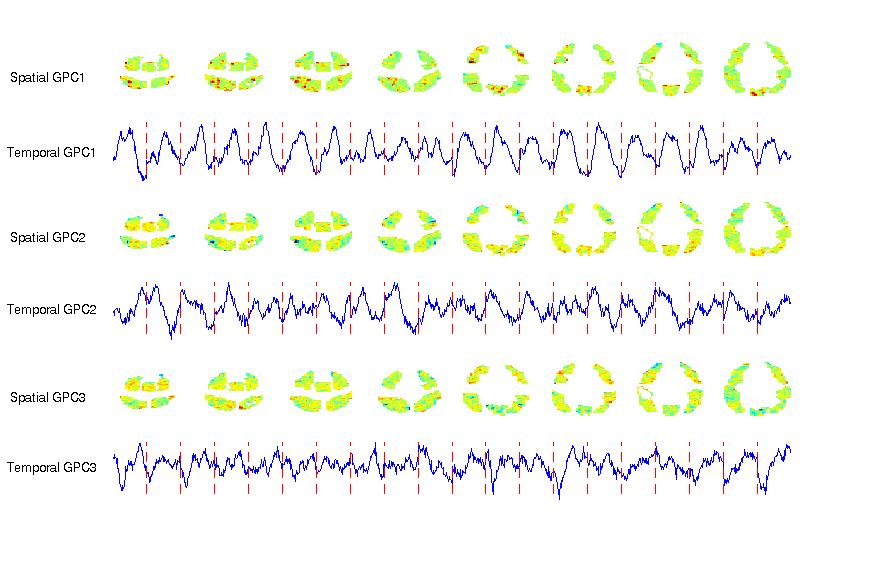}
\includegraphics[width=4.25in]{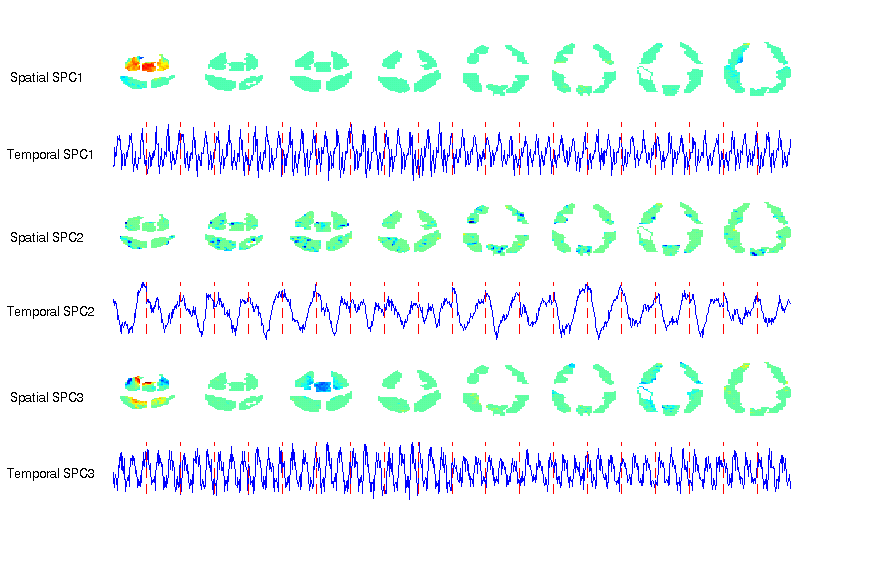}\includegraphics[width=4in]{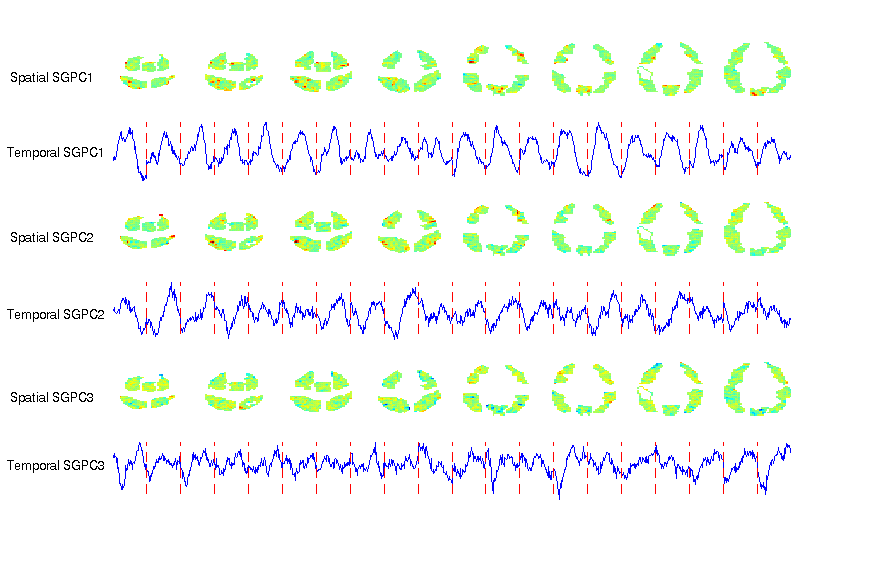}
\caption{\em \footnotesize Eight axial slices with corresponding time
  series for the first three PCs of the StarPlus fMRI data for PCA (top
  left), Sparse PCA (bottom 
  left), Generalized PCA (top right) and Sparse Generalized PCA
  (bottom right). Dotted red vertical lines in the time series denote
  the beginning and end of each task where an image was shown
  accompanied by an audible sentence that corresponded to the image.
  The first three spatial PCs of PCA and the first and third of Sparse
  PCA exhibit large patterns of spatial noise where the corresponding
  time series appear to be artifacts or scanner noise, unrelated to
  the experiment. 
  The temporal PCs of GPCA and Sparse GPCA, on the other hand, exhibit
  a clear pattern with respect to the experimental task.  The spatial
  PCs of GPCA, however are somewhat noisy, whereas when sparsity is
  encouraged, the spatial PCs of Sparse GPCA illustrate clear regions
  of activation related to the experimental task.  } 
\label{fig_fmri}
\end{sidewaysfigure}

\section{Discussion}
\label{section_discussion}

In this paper, we have presented a general framework for
incorporating structure into a matrix
decomposition 
and a two-way penalized matrix factorization.  Hence, we have
provided an alternative to the SVD, PCA and regularized PCA that is
appropriate for structured data.  We have also developed fast
computational algorithms to apply our methods to massive
data such as that of functional MRIs.  Along the way, we have provided
results that
clarify the types of penalty functions that may be employed on
matrix factors that are estimated via a deflation scheme.  Through
simulations and a real example on fMRI data, we have shown that GPCA
and regularized GPCA can be a powerful method for signal recovery and
dimension reduction of high-dimensional structured data.

While we have presented a general framework permitting heteroscedastic
errors in PCA and two-way regularized PCA, we currently only advocate the
applied use of our methodology for structured data.  Data in which
measurements are taken on a grid (e.g. regularly spaced time series,
spectroscopy, and image data) or on known Euclidean coordinates
(e.g. environmental data, epigenetic methylation data and unequally spaced
longitudinal data) have known structure which can be encoded by the
quadratic operators through smoothing or graphical operators as
discussed in Section \ref{section_interpret}.  Through close
connections to other non-linear unsupervised learning techniques and
interpretations of GPCA as a covariance decomposition and smoothing
operator,  the behavior our methodology for
structured data is well understood.  For data that is not
measured on a structured surface, however, more investigations need to
be done to determine the applicability of our methods.  In particular,
while it may be tempting to estimate both the quadratic operators,
via \citet{allen_2010_imp} for example, and the GPCA factors, this is
not an approach we advocate as separation of the noise and signal
covariance structure may be confounded.

For specific applications to structured data, however, there is much future
work to be done to determine how to best employ our methodology.
The utility of GPCA would be greatly enhanced by data-driven methods
to learn or estimate the optimal quadratic operators from certain
classes of structured matrices or an inferential approach to determining
the best pair of quadratic operators from a fixed set of options.
With structured data in particular, methods to achieve these are not
straightforward as the amount of variance explained or the matrix
reconstruction error is not always a good measure of how well the structural
noise is separated from the actual signal of interest (as seen in the
spatio-temporal simulation, Section \ref{section_simulations}).
Additionally as the definition of the ``signal'' varies from one
application to another, these issues should be studied carefully
for specific applications.  An area of future research would then be
to develop the theoretical properties of GPCA in terms of consistency or
information-theoretic bounds based on classes of signal and
structured noise.  
While these investigations are beyond the scope of this initial paper,
the authors plan on exploring these issues as well as applications to
specific data sets in future work.

There are many other statistical directions for future
research related to our work.  As many other classical multivariate
analysis methods are closely related to PCA, our framework for
incorporating structure and two-way regularization could be employed
in methods such as canonical correlations analysis, discriminant
analysis, multi-dimensional scaling, latent variable models, and
clustering.  Also several other statistical and machine learning techniques
are based on Frobenius norms which may be altered for structured data
along the lines of our approach.  Additionally, there are many areas
of research related to two-way regularized PCA models.  Theorem
\ref{gpmf_sol} elucidates the classes of penalties that can be
employed on PCA factors estimated via deflation that ensure
algorithmic convergence.  This 
convergence, however, is only to a local optimum. Methods are then
needed to find
good initializations, to estimate optimal penalty parameters, to find
the relevant range of penalty parameters comprising a full solution
path, and to
ensure convergence to a good solution.  Furthermore, asymptotic
consistency of several approaches to 
Sparse PCA has been established \citep{johnstone_spca_2004,
  amini_2009, ma_spca_2010}, but more work needs to be 
done to establish consistency of two-way Sparse PCA.

Finally, our methodology work has broad implications in a wide array
of applied disciplines.  Massive image data is common in areas of
neuroimaging, microscopy, hyperspectral imaging, remote sensing, and
radiology.  Other examples of high-dimensional structured data can be
found in environmental and climate studies, times series and finance,
engineering sensor and network data,
and astronomy.  Our GPCA and regularized GPCA methods can be used with
these structured data sets for improved dimension reduction, signal
recovery, feature selection, data visualization and exploratory data
analysis.

In conclusion, we have presented a general mathematical framework for
PCA and regularized PCA for massive structured data. Our methods have
broad statistical and 
applied implications that will lead to many future areas of research.

\section{Acknowledgments}

The authors would like to thank Susan Holmes for
bringing relevant references to our attention and Marina Vannucci for
the helpful advice in the preparation of this manuscript.
Additionally, we are grateful to the anonymous referees and associate
editor for helpful comments and suggestions on this work.  J. Taylor
was partially supported by NSF DMS-0906801, and L. Grosenick was
supported by NSF IGERT Award \#0801700.

\appendix

%% \section{The SVD \& PCA with Dependent Data}

%% \begin{figure}
%% \includegraphics[width=6.5in]{fig_mexample2_fpca_talk.jpg}
%% \end{figure}

\section{Proofs}

\begin{proof}[Proof of Theorem ~\ref{gmd_sol}]
{\footnotesize
We show that the GMD problem, \eqref{gmd}, at $\U^{*},\D^{*},\V^{*}$
is equivalent to the SVD problem, \eqref{svd}, for $\tilde{\X}$ at
$\tilde{\U},\tilde{\D}, \tilde{\V}$.  Thus, if
$\tilde{\U},\tilde{\D},\tilde{\V}$ minimizes the SVD problem,
\eqref{svd}, then 
$\U^{*},\D^{*},\V^{*}$ minimizes the GMD problem, \eqref{gmd}.
We begin with the objectives:
\begin{align*}
&|| \X - \U^{*} \D^{*} (\V^{*})^{T} ||_{\Q,\R}^{2} = \mathrm{tr} \left( \Q
   \X \R \X^{T} \right) - 2 \mathrm{tr}\left( \Q \U^{*} \D^{*}
   (\V^{*})^{T} \R \X^{T} \right)  \\
& \qquad \qquad + \mathrm{tr} \left( \Q \U^{*} \D^{*} (\V^{*})^{T} \R \V^{*} (\D^{*})^{T}
   (\U^{*})^{T} \right) \\ 
& \qquad = \mathrm{tr} \left(\tilde{\Q} \tilde{\Q}^{T}  \X \tilde{\R}
   \tilde{\R}^{T} \X^{T} \right)
 -2 \mathrm{tr} \left( \tilde{\Q}  \tilde{\Q}^{T}
   \tilde{\Q}^{-1} \tilde{\U} \tilde{\D} \tilde{\V}^{T}
   (\tilde{\R}^{-1})^{T}  \tilde{\R} \tilde{\R}^{T} \X^{T} \tilde{\Q}
   \right)   \\
& \qquad \qquad  + \mathrm{tr} \left( \tilde{\Q}  \tilde{\Q}^{T}  \tilde{\Q}^{-1}
   \tilde{\U} \tilde{\D} \tilde{\V}^{T} 
   (\tilde{\R}^{-1})^{T} \tilde{\R}
   \tilde{\R}^{T} \tilde{\R}^{-1} \tilde{\V} \tilde{\D}^{T}
   \tilde{\U}^{T} \tilde{\Q}^{-1}  \right)  \\ 
&\qquad = \mathrm{tr} \left( \tilde{\X} \tilde{\X}^{T} \right)
-2 \mathrm{tr} \left( \tilde{\U} \tilde{\D} \tilde{\V}^{T}
   \tilde{\X}^{T} \right)  + \mathrm{tr} \left( \tilde{\U} \tilde{\D}
   \tilde{\V}^{T} \tilde{\V} \tilde{\D}^{T} \tilde{\U}^{T}  \right)  
= || \tilde{\X} - \tilde{\U} \tilde{\D} \tilde{\V}^{T} ||_{F}^{2}.
\end{align*}
One can easily verify that the constraint regions are equivalent:
$(\U^{*})^{T} \Q \U^{*} = \tilde{\U}^{T} \tilde{\Q}^{-1} \Q
(\tilde{\Q}^{-1})^{T} \tilde{\U} = \tilde{\U}^{T} \tilde{\U}$ and
similarly for $\V$ and $\R$.  This completes the proof.  }
\end{proof}

\begin{proof}[Proof of Corollaries to Theorem ~\ref{gmd_sol}]
{\footnotesize
These results follow from the properties of the SVD
\citep{eckart_1936, horn_johnson, golub_van_loan_1996} and the
relationship between the GMD and SVD given in Theorem ~\ref{gmd_sol}.
For Corollary ~\ref{cor_k}, we use the fact that $\mathrm{rank}(\A \B)
\leq \mathrm{min} \{ 
  \mathrm{rank}(\A), \mathrm{rank}(\B) \}$ for any two matrices $\A$ and
  $\B$; equality holds if $range(\A) \cap null(\B) = \emptyset$.  Then, $\mathrm{rank}(\tilde{\X} ) = k =
  \mathrm{min}\{ \mathrm{rank}(\X), \mathrm{rank}(\Q),
  \mathrm{rank}(\R) \}$.}
\end{proof}

\begin{proof}[Proof of Proposition ~\ref{gmd_alg_sol}]
{\footnotesize
We show that the updates of $\uvec$ and $\vvec$ in the GMD Algorithm
are equivalent to the updates in the power method for computing the
SVD of $\tilde{\X}$.  
Writing the update for $\uvec_{k}$ in terms of$\tilde{\X}$, 
$\tilde{\uvec}$, and $\tilde{\vvec}$ (suppressing the index $k$), we have:
\begin{align*}
(\tilde{\Q}^{-1})^{T} \tilde{\uvec} &= \frac{( (\tilde{\Q}^{-1})^{T}
    \tilde{\X} \tilde{\R}^{-1}
) \R (\tilde{\R}^{-1})^{T} \tilde{\vvec} }{||( (\tilde{\Q}^{-1})^{T}
    \tilde{\X} \tilde{\R}^{-1} 
) \R (\tilde{\R}^{-1})^{T}  \tilde{\vvec}||_{\Q}  }, \Rightarrow
\tilde{\uvec} &= \frac{\tilde{\X} \tilde{\vvec} }{ \sqrt{
    \tilde{\vvec}^{T} \tilde{\X}^{T} \tilde{\Q}^{-1} \Q (\tilde{\Q}^{-1})^{T}
    \tilde{\X} \tilde{\vvec} }} = \frac{\tilde{\X} \tilde{\vvec} }{ ||\tilde{\X} \tilde{\vvec} ||_{2}}.
\end{align*}
Notice that this last form in $\tilde{\uvec}$ is that of the
power method for computing the SVD of $\tilde{\X}$
\citep{golub_van_loan_1996}.  A similar 
calculation for $\vvec$ yields an analogous form.  Therefore, the GMD
Algorithm is equivalent to the algorithm which converges to the SVD of
$\tilde{\X}$.}
\end{proof}

\begin{proof}[Proof of Proposition ~\ref{prop_gpca}]
{\footnotesize
We show that the objective and constraints
in \eqref{gpca} at the GMD solution are the same as that of the rank
$k$ PCA problem
for $\tilde{\X}$.  (For simplicity of notation, we suppress
the index $k$ in the calculation.)
\begin{align*}
\vvec^{T} \R \X^{T} \Q \X \R \vvec &= \tilde{\vvec}^{T}
(\tilde{\R}^{-1})^{T} \R \tilde{\R}^{-1} \tilde{\X}
(\tilde{\Q}^{-1})^{T} \Q \tilde{\Q}^{-1} \tilde{\X}
(\tilde{\R}^{-1})^{T} \R \tilde{\R}^{-1} \tilde{\vvec} =
\tilde{\vvec}^{T} \tilde{\X}^{T} \tilde{\X} \tilde{\vvec} \\
\vvec^{T} \R \vvec &= \tilde{\vvec}^{T} (\tilde{\R}^{-1})^{T} \R
\tilde{\R}^{-1} \tilde{\vvec} = \tilde{\vvec}^{T} \tilde{\vvec}.
\end{align*}
Then, the left singular vector, $\tilde{\vvec}$, that maximizes
the PCA problem, is the same as the left GMD factor that maximizes
\eqref{gpca}.} 
\end{proof}

\begin{proof}[Proof of Corollary ~\ref{cor_gpca_var_ex}]
{\footnotesize
This result follows from the relationship of the GMD and GPCA to the
SVD and PCA.  Recall that for PCA, the amount of variance explained by
the $k^{th}$ PC is given by $d_{k}^{2} / || \X ||_{F}^{2}$ where
$d_{k}$ is the $k^{th}$ singular value of $\X$.  Then, notice that the
stated result is equivalent to the proportion of variance explained by
the $k^{th}$ singular vector of $\tilde{\X}$:  $\tilde{d}_{k}^{2} /
\mathrm{tr}( \tilde{\X} \tilde{\X}^{T} ) = d_{k}^{2} / \mathrm{tr}( \Q
\X \R \X^{T} )$.  
}\end{proof}

\begin{proof}[Proof of Theorem ~\ref{gpmf_sol}]
{\footnotesize
We will show that the result holds for $\vvec$, and the argument for
$\uvec$ is analogous.  Consider optimization of
\eqref{gpmf} with respect 
to $\vvec$ with $\uvec$ fixed at $\uvec'$.  The problem 
is concave in $\vvec$ and there exists a strictly feasible point,
hence Slater's condition holds and the 
KKT conditions are necessary and sufficient for optimality.  Letting
$\y = \X^{T} \Q \uvec$, these are given by:
$\R \y - \lamv \nabla P_{1}( \vvec^{*} ) - 2 \gamma^{*} \R \vvec^{*} =
0$ and $\gamma^{*}( (\vvec^{*})^{T} \R \vvec^{*} - 1) = 0$.
Here, $\nabla P_{1}( \vvec^{*} )$ is a subgradient of $P_{1}()$ with
respect to $\vvec^{*}$. 
Now, consider the following solution to the penalized regression problem:
$\hat{\vvec} = \argmin_{\vvec} \{ \frac{1}{2}|| \y - \vvec ||_{\R}^{2} + \lamv
P_{1}(\vvec)  \}$. 
This solution must satisfy the subgradient conditions.  Hence,
$\forall \ c > 0$, we have:
\begin{align*}
0 = \R \y - \lamv \nabla P_{1}( \hat{\vvec} ) - \R \hat{\vvec} = \R \y  - \lamv \nabla P_{1}( \hat{\vvec} ) - \R( c \hat{\vvec})
\frac{1}{c} = \R \y  - \lamv \nabla P_{1}( \tilde{\vvec} ) - \R \tilde{\vvec} \frac{1}{c},
\end{align*}
where $\tilde{\vvec} = c \hat{\vvec}$.  The last step follows because
since $P_{1}()$ is convex and homogeneous of order one, $ \nabla P_{1}(x) =
\nabla P_{1}(cx) \ \ \forall \ c > 0$.  

Let us take $c = 1 / || \hat{\vvec} ||_{\R}$; we see that this
satisfies $c > 0$.  Then, letting $\gamma^{*} = \frac{1}{2c} = ||
\hat{\vvec} ||_{\R} / 2$, we see that the subgradient condition of
\eqref{gpmf} is equivalent to that of the penalized
regression problem.  Putting these together, we see that the pair $(
\vvec^{*} = \hat{\vvec } / || \hat{\vvec} ||_{\R}, \gamma^{*} = ||
\hat{\vvec} ||_{\R} / 2)$ satisfies the KKT conditions and is hence
the optimal solution.  Notice that $\gamma^{*} = 0$ only if $||
\hat{\vvec} ||_{\R} = 0$.  From our discussion of the quadratic norm,
recall that this can only occur if $\hat{\vvec} \in null(\R)$ or
$\hat{\vvec} = 0$.  Since we exclude the former by assumption, this
implies that if $\hat{\vvec} = 0$, then $\gamma^{*} = 0$ and the
inequality constraint in \eqref{gpmf} is not tight.  In this case, it
is easy to verify that the pair $(\vvec^{*} = 0, \gamma^{*} = 0)$
satisfy the KKT conditions.  Thus, we have proven the desired result.   }
\end{proof}

\begin{proof}[Proof of Claim ~\ref{gpmf_sparse_diag}]
{\footnotesize
Examining the subgradient condition, we have that $\hat{\vvec} = \y -
\lambda \R^{-1} 
\mathbf{1}_{(p)} \nabla P( \hat{\vvec})$.  This
can be written as such because the subgradient equation is separable
in the elements of $\hat{\vvec}$ when $\R$ is diagonal.  Then, it
follows that  soft 
thresholding the elements of $\y$ with the penalty vector $\lambda
\R^{-1} \mathbf{1}_{(p)}$ solves the subgradient equation.  } 
\end{proof}

\begin{proof}[Proof of Claim ~\ref{gpmf_sparse_gen}]
{\footnotesize
We can write the problem, $\frac{1}{2} || \y - \vvec
||_{\R}^{2} + \lamv || \vvec ||_{1}$, as a lasso regression problem in
terms of $\tilde{\R}$, $\frac{1}{2} || \tilde{\R}^{T} \y -
\tilde{\R}^{T} \vvec ||_{2}^{2} + \lamv || \vvec ||_{1}$. 
\citet{coor_descent_2007} established that the lasso regression
problem can be solved by iterative coordinate-descent.  Then,
optimizing with respect to each coordinate, $\vvec_{j}$:
$\hat{\vvec}_{j} =  \frac{1}{\tilde{\R}_{rj}^{T} \tilde{\R}_{rj}} S
\left(  \tilde{\R}_{rj}^{T} \tilde{\R} \y  -
\tilde{\R}_{rj}^{T} \tilde{\R}_{r, \neq j} \hat{\vvec}_{\neq j} ,
\lambda \right)$.  
Now,
several of these quantities can be written in terms of $\R$ so that
$\tilde{\R}$ does not need to be computed and stored:
$\tilde{\R}_{rj}^{T} \tilde{\R}_{rj} = \R_{jj}$,  
$\tilde{\R}_{rj}^{T} \tilde{\R} = \R_{rj}$, and 
$\tilde{\R}_{rj}^{T}  \tilde{\R}_{r, \neq j} = \R_{j,\neq j}$.
Putting these together, we have the desired result.  }
\end{proof}

\begin{proof}[Proof of Proposition ~\ref{prop_gpmf_smooth}]
{\footnotesize
We will show that Algorithm ~\ref{alg_smooth} uses a generalized
gradient descent method converging to the minimum of the $\Omeg$-norm
generalized regression problem.
First, since $\Bmat$ is full rank, there is a one to one mapping
between $\vvec$ and $\wvec$.  Thus, any $(\wvec^{*},\eta^*)$ minimizing $f(\wvec,\eta) =
\frac{1}{2} || \y - \Omeg^{-1/2} \wvec - \mathbf{N}\eta ||_{\R}^{2} + \lamv || \wvec ||_{2}$
yields  $\vvec^{*} = \Bmat \begin{pmatrix} \wvec^{*} \\ \eta^* \end{pmatrix}$ 
which minimizes the $\Omeg$-norm
generalized regression problem.

Now, let us call the first and second terms in  $f(\wvec,\eta)$, $g(\wvec,\eta)$ and
$h(\wvec,\eta)$ respectively.  Then, $h(\wvec,\eta)$ and $g(\wvec,\eta)$ are
obviously convex and $g(\wvec,\eta)$ is Lipschitz with a Lipschitz constant upper bounded by
$\|\Bmat^T \R \Bmat\|_{op}$.  To see this, note that 
$\nabla g(\wvec,\eta) = - \Bmat^{T} \R \left( \y -
\Bmat \begin{pmatrix}\wvec \\ \eta \end{pmatrix} \right)$ is linear in $(\wvec, \eta)$.  Therefore, we can take
the operator norm of the linear part as a Lipschitz constant.  We also
note that $f(\wvec \eta)$ is separable in $\wvec$ and $\eta$ and hence
block-wise coordinate descent can be employed.

 Putting all of these together, the
generalized gradient update for $\wvec$ and the update for $\eta$ that
converge to the minimum 
of the $\Omeg$-norm penalized regression problem is given by the
following updates \citep{vandenberghe, beck_fista_2009}:  
$$
\wvec^{(t+1)} = \argmin_{\wvec} \left\{ \frac{L}{2} \left\| \wvec - \left(
\wvec^{(t)} + \frac{1}{L} \Bmat^{T} \R \left( \y - \Omeg^{-1/2} \wvec^{(t)} - \mathbf{N}\eta^{(t)} \right)
\right) \right\|_{2}^{2} + \lamv || \wvec ||_{2}   \right\} 
$$
$$
\eta^{(t+1)} = \argmin_{\eta} f(\wvec^{(t+1)}, \eta ).
$$

Notice that the first minimization problem is of the form of the group
lasso regression problem with an identity predictor matrix.  The
solution to this is problem given in \citet{yuan_2006_group} and is the
update in Step (b) of Algorithm \ref{alg_smooth}.  The second
minimization problem is a simple linear regression and is the update
in Step (c).  This
completes the proof.} 
\end{proof}

\begin{proof}[Proof of Proposition ~\ref{prop_var_ex}]
{\footnotesize
This result follows from an argument in \citet{shen_spca_2008}.
Namely, since our regularized GPCA formulation does not impose
orthogonality of subsequent factors in the $\Q,\R$-norm, the effect of
the first $k$ factors must be projected out of the data to compute the
cumulative proportion of variance explained.  \citet{shen_spca_2008}
showed that for Sparse PCA with a penalty on $\vvec$, the total
variance explained is given by $\mathrm{tr}( \X_{k} \X_{k}^{T} )$
where $\X_{k} = \X \Pmat_{k}^{(\V)} = \X \V_{k} ( \V_{k}^{T} \V_{k})^{-1}
\V_{k}^{T}$ with 
$\V_{k} = [\vvec_{1} \ \ldots \ \vvec_{k} ]$.  When penalties are
incorporated on both factors, one must define $\X_{k} = \Pmat_{k}^{(\U)}
\X \Pmat_{k}^{(\V)}$.  We show that the total proportion of variance
explained (the numerator of our stated result) is equivalent to
$\mathrm{tr}( \tilde{\X}_{k} \tilde{\X}_{k})$, where $\tilde{\X}_{k}$
is equivalent to that as defined in Theorem ~\ref{gmd_sol}.  First,
notice that $\Pmat_{k}^{(\U)} = \tilde{\Q}^{-1} \tilde{\U}_{k} (
\tilde{\U}_{k}^{T} ( \tilde{\Q}^{-1})^{T} \Q \tilde{\Q}^{-1}
\tilde{\U}_{k} )^{-1} \tilde{\U}_{k}^{T} ( \tilde{\Q}^{-1})^{T} =
\tilde{\Q}^{-1} \Pmat_{k}^{(\tilde{\U})} ( \tilde{\Q}^{-1} )^{T}$ and
equivalently, $\Pmat_{k}^{(\V)} = \tilde{\R}^{-1}
\Pmat_{k}^{(\tilde{\V})} ( \tilde{\R}^{-1} )^{T}$.  Then, $\X_{k} =
\Pmat_{k}^{(\U)} \Q \X \R \Pmat_{k}^{(\V)} = \tilde{\Q}^{-1}
\Pmat_{k}^{(\tilde{\U})} \tilde{\X} \Pmat_{k}^{(\tilde{\V})}
\tilde{\R}^{-1}$, and we have that \\
$\mathrm{tr}( \Q \X_{k} \R
\X_{k}^{T} ) = \tilde{\Q} \tilde{\Q}^{T} \tilde{\Q}^{-1}
\Pmat_{k}^{(\tilde{\U})} \tilde{\X} \Pmat_{k}^{(\tilde{\V})}
\tilde{\R}^{-1} \tilde{\R} \tilde{\R}^{T} ( \tilde{\R}^{-1})^{T}
\Pmat_{k}^{(\tilde{\V })} \tilde{\X}^{T} \Pmat_{k}^{(\tilde{\U})} (
\tilde{\Q}^{-1})^{T} = \mathrm{tr}( \tilde{\X}_{k} \tilde{\X}_{k}^{T}
)$, thus giving the desired result.
}\end{proof}

\begin{proof}[Proof of Claim ~\ref{claim_bic}]
{\footnotesize
\citet{lee_ssvd_2010} show that estimating $\vvec$ in the manner
described in our GPMF framework is analogous to a penalized regression
problem with a sample size of $np$ and $p$ variables.  Using this and
assuming that the noise term of the penalized regression 
problem is unknown, we arrive at the given form of the BIC.  Notice
that the sums of squares loss function is replaced by the $\Q,\R$-norm
loss function of the GMD model.   } 
\end{proof}

{\scriptsize
\bibliographystyle{Chicago}
\bibliography{fmri}
}

\end{document}